\documentclass[lettersize,journal]{IEEEtran}
\usepackage{amsmath,amsfonts}
\usepackage{amssymb}
\usepackage{algorithmic}
\usepackage{algorithm}
\usepackage{array}
\usepackage{units}
\usepackage[caption=false,font=footnotesize,labelfont=rm,textfont=rm]{subfig}
\usepackage{textcomp}
\usepackage{stfloats}
\usepackage{url}
\usepackage{verbatim}
\usepackage{graphicx}
\usepackage{cite}
\usepackage{float}  
\usepackage{lipsum}
\usepackage{amsthm}
\usepackage{xcolor}
\usepackage{diagbox}
\usepackage{hyperref}
\usepackage{cleveref}
\usepackage[hyphenbreaks]{breakurl}
\usepackage{epstopdf}
\usepackage{xr}
 
\newtheorem{theorem}{Theorem}

\newtheorem{lemma}{Lemma}
\newtheorem{definition}{Definition}

\newcommand\figref{Fig.~\ref}

\begin{document}

\title{Online Signed Sampling of Bandlimited Graph Signals}
\author{Wenwei Liu, Hui Feng \IEEEmembership{Member,~IEEE,} Feng Ji, Bo Hu \IEEEmembership{Member,~IEEE}
\thanks{

  Wenwei Liu, is with the School of Information Science and Technology, Fudan University, Shanghai 200433, China (e-mail: wwliu21@m.fudan.edu.cn).
  
 Hui Feng and Bo Hu are with the School of Information Science and Technology, Fudan University, Shanghai 200433, China, and also with Shanghai Institute of Intelligent Electronics \& Systems, Shanghai 200433, China (e-mail: hfeng@fudan.edu.cn; bohu@fudan.edu.cn).

 Feng Ji is with the School of Electrical and Electronic Engineering, Nanyang Technological University, Singapore (e-mail: jifeng@ntu.edu.sg).

A preliminary work has been published in ICASSP 2022 \cite{liu2022recovery}.

\copyright 2024 IEEE. Personal use of this material ispermitted.Permission from IEEE must beobtained for all other uses, in any current orfuture media, including reprinting/republishing this material for advertising orpromotional purposes,creating new collectiveworks, for resale or redistribution to serversor lists, or reuse of any copyrightedcomponent of this work in other works.

Digital Object Identifier 10.1109/TSIPN.2024.3356794
}
}

\markboth{Journal of \LaTeX\ Class Files}%
{Shell \MakeLowercase{\textit{et al.}}: A Sample Article Using IEEEtran.cls for IEEE Journals}


\maketitle

\begin{abstract}
The theory of sampling and recovery of bandlimited graph signals has been extensively studied. 
However, in many cases, the observation of a signal is quite coarse. For example, users only provide simple comments such as ``like'' or ``dislike'' for a product on an e-commerce platform. This 
is a particular scenario where only the sign information of a graph signal can be measured. 
In this paper, we 
are interested in how to sample 
based on sign information in an online manner, by which the direction of the original graph signal can be estimated. 
The online signed sampling problem of a graph signal can be formulated as a Markov decision process in a finite horizon. Unfortunately, it is intractable for large size graphs. 
We propose a low-complexity greedy signed sampling algorithm (GSS) as well as a stopping criterion.
Meanwhile, we prove that the objective function is adaptive monotonic and adaptive submodular, so that the performance is close enough to the global optimum with a lower bound. 
 Finally, we demonstrate the effectiveness of the GSS algorithm by both synthesis and realworld data. 

\end{abstract}

\begin{IEEEkeywords}
Graph signal processing, sign information, Markov decision process, online sampling.
\end{IEEEkeywords}

\section{Introduction}
\label{intro}
\IEEEPARstart{I}{n} practice, data is often distributed on irregular topologies, such as transportation networks, social networks, and sensor networks \cite{stankovic2019understanding,zhou2004regularization,zhang2020deep,zhang2024graph}. In order to analyze and process such data in view of associated network topology, graph signal processing (GSP) has gained increasing popularity in recent years. It has been used in community mining, computer vision, chemical and pharmaceutical industries, and other fields\cite{shuman2013emerging,stankovic2020data,xia2021graph}. Many GSP concepts are analogous to their counterparts in traditional signal processing theory\cite{cvetkovic2009applications,gavili2017shift}. For example, like the traditional Fourier transform (FT), the graph Fourier transform (GFT) establishes a relationship between the vertex-domain and the frequency-domain of graph signals.

Sampling of graph signals is an important and fundamental research topic in GSP \cite{wang2015generalized,ortega2018graph}, whose purpose is to find a subset of vertices sufficient to give full knowledge of graph signals with prescribed properties. In general, sampling of graph signals can be divided into two categories. One is deterministic sampling, which seeks to find a sampling set that minimizes a predetermined metric function. The other is random sampling, which examines the importance of each vertex and assigns a corresponding sampling probability\cite{tanaka2020sampling,tsitsvero2016signals}. 
In many cases, random sampling needs more samples than deterministic sampling to achieve the desired accuracy in subsequent tasks\cite{tanaka2020sampling}.

Sampling of graph signals can also be categorized into offline (batched) or online (sequential) manners. 
Offline sampling solves the sampling set altogether according to the prior information, and then conducts subsequent observations on this sampling set. Related methods include \cite{marques2015sampling,xie2017design}.
Sampling in an offline manner relies heavily on the prior knowledge and designs the sampling set in advance using only the information of graph topology and signal prior. 
Online sampling selects samples one by one, which makes full use of the information from historical observations during the sampling process. Related works can be found in Section~\ref{related:online}. In contrast, online sampling takes more information into account for timely evaluation of the 
unsampled vertices and makes more explicit choices, which may result in a smaller sample size.

Signal recovery from sampled data is also a widely studied problem \cite{stankovic2019understanding}, and our focus in this paper revolves around reconstructing the signal's direction, not the entire signal. 
The ability to determine signal direction holds significant relevance in numerous practical scenarios. For example, when the amplitude of the signal in some situations is available, the entire signal can be deduced by knowing only the direction of the signal. A typical example is that in a power grid system, where the total load can be measured, estimating the direction of the signal becomes pivotal to knowing the specific signal value assigned to each vertex \cite{fang2011smart}. Additionally, there exist scenarios where the signal's amplitude bears little importance. For example, in a commodity network, if the purpose is to determine the ranking of commodity scores, then the signal amplitude is less relevant, and only estimating the signal direction is sufficient \cite{falk2019practical}. 

In many practical scenarios, we only have coarse information about the signal of interest. 
For example, in a goods rating system, suppose we model the collection of all goods by a graph, where edges connect similar goods. The scores of the goods are viewed as the graph signal \cite{falk2019practical}. It is usually difficult for customers to give accurate scores. Therefore, customers usually only provide simple comments on the products after consumption, such as ``thumbs up'', ``thumbs down'' and ``indifference''. 

Moreover, making comparisons is much easier than rating. 
If customers are asked to choose between A and B in terms of preference, then ``A is better'', ``B is better'' and ``A is about the same as B'' correspond to the score of A being higher, lower and the same as that of B. In the above scenario, the only information we have is 
whether the signal of a vertex is higher than 
 that of a neighbor.

 Such quantized information above mentioned are two examples of ``sign information''. 
It will be useful to make reliable product recommendations 
if the signal can be estimated from such partial sign information. 
Motivated by the example, in this paper, we explore how to perform sampling and recover the direction of 
a bandlimited graph signal based on sign information of partial vertices and edges.

Sign information of vertices and edges reflects the signal polarity on each vertex and the ordinal relations of signal magnitude in each local neighborhood. On the other hand, bandlimitedness encodes the variation pattern of the signal. The combination of sign information and bandlimitedness allows us to obtain a qualified estimation of an unknown signal when the accuracy of the original signal is not very demanding.

Our main contributions are two folds. On the one hand, 
we provide the performance analysis of the GSS algorithm by giving a lower bound on model performance. 
On the other hand, we develop a stopping criterion of the GSS algorithm to avoid redundant samples.

The rest of this paper is organized as follows. In Section~\ref{sec:related}, 
we review some related work, mainly about: the recovery of 1-bit quantized signals and online sampling of graph signals. 
In Section~\ref{sec:model} and Section~\ref{sec:process}, we describe the model setup and formulate the problem. 
In Section~\ref{sec:Algorithm}, we describe the GSS algorithm and the recovery algorithm called unital projection onto convex sets (UPOCS). 
In Section~\ref{theory}, we make theoretical discussions. We interpret the online sampling process of a graph signal as a Markov decision process (MDP). Taking the size of the feasible region as the objective function, we prove its adaptive monotonicity and adaptive submodularity, which leads to a performance guarantee. We provide the stopping criterion of the GSS algorithm, with theoretical analysis. We verify the effectiveness of the GSS algorithm with both synthetic and real datasets in Section~\ref{sec:experiments} and conclude in Section~\ref{sec:conclusion}.

In order to present most frequently used concepts and corresponding symbols more intuitively, we provide a notation table, as shown in Table \ref{tab:notation}. 

\begin{table}[htbp]
\normalsize
	\centering
	\caption{NOTATIONS}
	\begin{tabular}{l l}  
			\hline 
			Notation                & Description         \\
			\hline
            $\boldsymbol{x}$        & a bandlimited graph signal\\
            $B$ & bandwidth \\
            $\boldsymbol{h}$        & frequency coefficients\\
            $\boldsymbol{y}$        & sign information of samples\\
            $\boldsymbol{U}_B$      & the GFT bases indexed by the passband \\
            $\mathbf{\Psi}$  & the joint sampling matrix \\
            $\mathcal{S}$  &  a sampling set  \\
            $a_{0:t}$  & the sampling sequence from time $0$ to $t$ \\
            $y_{0:t}$  & the sign information sequence of $a_{0:t}$ \\
            $\mathcal{A}_t$  &  the set of all unsampled vertices and edges at \\
              & time $t$ \\
            $O_t$  &  the observation sequence at time $t$  \\
            $\mathcal{J}_t$   &  the feasible region at time $t$\\
            $\hat{\mathcal{J}}_t$  &  the convex hull of $\mathcal{J}_t \cup \{\boldsymbol{0}\}$, whose volume is \\ 
            & used to measure the size of $\mathcal{J}_t$    \\
            $\mathcal{H}_{a_t}$  & a hyperplane determined by the sample $a_t$ \\
            $\text{Vol}(\cdot)$  &  the volume of the given body \\
            $\boldsymbol{\psi}_{a_t}^\top$   &  a row of $\boldsymbol{\Psi}$ corresponding to the sample $a_t$   \\
            $\Phi(\cdot)$  &  the expected reward of a sampling sequence  \\
            $\mathcal{Z}$  & the set of EVs \\
            $\Delta(\cdot \mid \cdot)$    & the conditional expected marginal benefit of a \\ 
             & sample given an observation sequence \\
            $\boldsymbol{P}$ & a projection matrix \\
            $\delta$        & the (average) error in angle \\
			\hline
	\end{tabular}
	\label{tab:notation}
\end{table}

\section{Related Works}
\label{sec:related}
\subsection{Recovery of Signals from 1-bit Quantization}
\label{ssec:1-bit recovery}
Graph signal recovery using sign information is related to the recovery of 1-bit quantized signals. Related literature can be divided into two categories, random sensing, and direct observation, according to the type of sensing matrix. 

For the case of random sensing, combined with compressed sensing, the bound of the required number of measurements was investigated, with which universal exact recovery is possible for sparse signals with bounded dynamic range \cite{bansal2022universal}.
Beheshti \textit{et al}. studied how to set the best adaptive threshold for dictionary-sparse signals, and proposed an iterative optimization method for recovery \cite{beheshti2022adaptive}. Jacques \textit{et al}. propounded a binary iterative hard threshold algorithm and discussed recovery error bounds for spectral sparse signals in the ideal case \cite{jacques2013robust}. 
Deep learning methods have also been advanced in recent years, for example, Zeng \textit{et al}. proposed a deep unfolded network for the case where there is no priori information about the sensing matrix to jointly optimize the sensing matrix and the signal \cite{zeng2022one}. Tachella and Jacques formulated a self-supervised framework that uses operator consistency as the pretext, combined with measurement consistency to jointly design the loss function \cite{tachella2023learning}. 
Moreover, when dealing with complex signals, there is a special class of problems, 1-bit phase retrieval, which only considers whether the amplitude of the signal exceeds the threshold. The effect of the number of measurements on recovery is considered and a sub-conjugate gradient method to iteratively recover the signal of interest is presented \cite{eamaz2022one}. Kishore and Seelamantula determined the threshold according to the statistical prior knowledge and proposed the Wirtinger Flow recovery algorithm based on the gradient descent method \cite{kishore2020wirtinger}. 

For direct observation, the signal of interest at different time steps can be estimated by ML approach using adaptive thresholding \cite{khobahi2018signal}. Goyal and Kumar considered the AWGN dithering and selected the vertices with the largest contribution to the signal energy, then a recursive recovery algorithm was developed which theoretically proved that the recovery error is bounded by the contraction mapping theorem \cite{goyal2018estimation}. As we can see, these studies consider dithering or adaptive thresholds with some prior knowledge to provide additional information for signal recovery.

As for the recovery using sign information on edges, it can be traced back to the studies on recovering signals from zero crossings. 
The zero crossings of a signal, which indicate sign changes of the signal value between adjacent sampling points, reflect the oscillatory property of the signal. Logan proved in 1977 that the direction vector of a traditional discrete signal can be uniquely determined by its zero crossings if the bandwidth is less than one octave \cite{logan1977information}. Since then, plenty of subsequent studies have extended this conclusion and applied it to the field of image reconstruction in practice, such as \cite{rotem1986image}. 
Similarly, the phase of a signal implicitly indicates the variation pattern of a signal within its neighborhood\cite{oppenheim1981importance}. In \cite{dhar2019tests}, phase parameters of chirp signals are determined by means of hypothesis testing, and theoretical results are given for finite samples and asymptotics, respectively. In \cite{wu2008subspace}, a novel subspace-based algorithm is proposed to decompose the estimation of polynomial phases into subproblems and to estimate the polynomial coefficients in conjunction with a multiple signal classification approach.
In contrast, the sign information on edges studied in this paper is more intuitive than phase and more universal than zero crossings, generalizing to non-Euclidean spaces. 

The above works are mainly for traditional discrete signals except \cite{goyal2018estimation}. However, only the sign information on vertices is considered in \cite{goyal2018estimation}. 
Different from the above-mentioned works, we consider the recovery problem of unital bandlimited graph signals, with sign measurements on both vertices and edges.
At the same time, the impact of sampling on recovery is also considered.

\subsection{Online Sampling of Graph Signals}
\label{related:online}

Online sampling is to select elements from a set sequentially, where at each sampling step, we determine the next sample based on known samples and historical observations.
In GSP, online sampling is widely used and always efficient to solve the problem of sampling continuous graph signals. 

A classical approach is based on minimizing the recovery error covariance, including the average error and the worst-case error \cite{chamon2016near,wu2018greedy}. These methods are theoretically supported by submodularity or approximate submodularity. 
In addition, the algorithms of maximizing cutoff frequency based on the sampling theorem of bandlimited graph signals are also studied, which are often greedily implemented to avoid high complexity \cite{anis2014towards,gadde2015probabilistic}. Accordingly, some simplified versions have been proposed, such as \cite{ jayawant2018distance,tzamarias2018novel}. 

There are also some approaches that treat graph signals as stochastic models, often with certain probability assumptions, such as active sampling algorithms based on uncertainty criterion. These methods sequentially find the desired samples by adding the vertex on which the signal is the most uncertain \cite{lin2019active}, or the vertex that may cause the largest change in the existing estimation to the sampling sequence step by step\cite{berberidis2017active}. Although active sampling may be intuitively more efficient, it is usually difficult to find theoretical support.

To the best of our knowledge, there is rare research concerning online sampling based on sign information. 
Our previous work \cite{liu2022recovery} did a preliminary exploration without theoretical analysis.
Now we generalize this problem to jointly signed sampling of vertices and edges, and not only propose an improved online decision algorithm, but also provide theoretical justification using the framework of adaptive submodularity.

\section{Graph signals and sign information}
\label{sec:model}
In this section, we first describe the basic setup of the central problem and then introduce a key concept to study the problem.

\subsection{System Model}
\label{ssec: system model}
Consider a connected and undirected graph $\mathcal{G}=\{\mathcal{V},\mathcal{E},\boldsymbol{W}\}$, where $\mathcal{V}=\{v_1,v_2,\dots,v_N\}$ is the set of vertices, $\mathcal{E}$ is the set of edges, $\boldsymbol{W}\in \mathbb{R}^{N\times N}$ is the weighted adjacency
matrix with $\boldsymbol{W}_{ij} >0$ 
if $v_i$ and $v_j$ are connected, and $0$ otherwise. The degree of a vertex $v_i$ is $\mathit{d_i} = \sum_{j}\boldsymbol{W}_{ij}$, and the degree matrix $\boldsymbol{D}$ is defined as $\boldsymbol{D} = \text{diag}(d_1,d_2,\dots,d_N)$. 

Recall that a graph signal can be described by a vector $\boldsymbol{x} \in \mathbb{R}^N$ with the $i$-th component corresponding to the value assigned to $v_i$. 
The signal $\boldsymbol{x}$ can be transformed to the frequency-domain by the GFT: $\bar{\boldsymbol{x}} = \boldsymbol{U}^\top\boldsymbol{x}$, and $\bar{\boldsymbol{x}}$ can be transformed back to vertex-domain by the IGFT: $\boldsymbol{x} = \boldsymbol{U}\bar{\boldsymbol{x}}$. 
$\boldsymbol{U}$ can be any orthonormal matrix, and a common choice is the eigenvector matrix of a graph shift operator, such as the graph Laplacian $\boldsymbol{L} = \boldsymbol{D}-\boldsymbol{W}$ and the normalized Laplacian matrix $\boldsymbol{L}_n = \boldsymbol{D}^{-\frac{1}{2}}(\boldsymbol{D}-\boldsymbol{W})\boldsymbol{D}^{-\frac{1}{2}}$ \cite{shuman2013emerging}.

A graph signal $\boldsymbol{x}$ is bandlimited if there exist indices $f_1 < f_2 < \dots < f_B \leq N\ $ such that $\bar{\boldsymbol{x}}_k = 0$ 
for all $k \not \in \{f_1,f_2,\dots,f_B\}$. The support of $\bar{\boldsymbol{x}}$, i.e., $\{f_1,f_2,\dots,f_B\}$, is the passband of $\boldsymbol{x}$, and $B$ is the corresponding bandwidth. It can be verified that such a bandlimited graph signal can be represented as $\boldsymbol{x} = \boldsymbol{U}_B \boldsymbol{h}$, where $\boldsymbol{U}_B$ is the $N \times B$ submatrix of $\boldsymbol{U}$ with columns indexed by the passband, and $\boldsymbol{h}$ is the corresponding frequency coefficients.

As we all know, for a graph signal $\boldsymbol{x}\in \mathbb{R}^N$, the signal value on vertex $v_i$ is the $i$-th element of $\boldsymbol{x}$. 
Moreover, we also introduce the notion of the signal value on an edge $e_i = (v_p, v_q)$, which is defined as the difference of the $p$-th and $q$-th element of $\boldsymbol{x}$, i.e., $x_p - x_q$. 
For example, in a goods rating system where the products are viewed as the vertices and the scores are viewed as the graph signal, the signal value on an edge can reflect the difference in customer preference for the two products. 
In the paper, we specify that the difference on the vertex-pair is made by subtracting the component of the larger subscript from the component of the smaller subscript. The signal value on all the edges can be expressed as
\begin{equation*}
    \boldsymbol{x}_\mathcal{E} = \boldsymbol{\Xi}\boldsymbol{x},
\end{equation*}
where $\boldsymbol{\Xi} \in \mathbb{R}^{|\mathcal{E}|\times N}$ is the incidence matrix with entries $\Xi_{ip} = 1$, $\Xi_{iq} = -1$ if the $i$-th edge is $e_i=(v_p,v_q)$, and $0$ otherwise. 

Sampling of $M_1$ vertices can be described as follows
\begin{equation*}
    \boldsymbol{x}_v = \boldsymbol{\Psi}_\mathit{v} \boldsymbol{x}.
\end{equation*}
In the above formula, $\boldsymbol{\Psi}_\mathit{v}\in\mathbb{R}^\mathit{M_1\times N}$ is the sampling matrix in vertex-domain, with entries $(\boldsymbol{\Psi}_\mathit{v})_\mathit{ij}=1$ if the $\mathit{i}$-th sample is vertex $v_\mathit{j}$, and $0$ otherwise. $\boldsymbol{\Psi}_\mathit{v}$ can be considered as the result of row selection on the identity matrix $\boldsymbol{I}_N$.
Sampling of $M_2$ edges can be defined similarly, which can be described as 
\begin{equation*}
    \boldsymbol{x}_e = \boldsymbol{\Psi}_\mathit{e} \boldsymbol{x}.
\end{equation*}
In the above formula, $\boldsymbol{\Psi}_e \in \mathbb{R}^{M_2\times N}$ has entries $(\boldsymbol{\Psi}_\mathit{e})_{ip} = 1$, $(\boldsymbol{\Psi}_\mathit{e})_{iq} = -1$ if the $i$-th sample is edge $(v_p,v_q)$, and $0$ otherwise. Similarly, $\boldsymbol{\Psi}_\mathit{e}$ can be considered as the result of row selection on the incidence matrix $\boldsymbol{\Xi}$. Intuitively, sampling in the vertex-domain is to pick the elements in $\boldsymbol{x}$ to form the sampled signal, while sampling in the edge-domain is to pick the elements in $\boldsymbol{x}_\mathcal{E}$ to form the sampled signal.

In this paper, we consider the joint sampling of vertices and edges, which means
that sampling requires one to select subsets of vertices and edges
to form the sampling set $\mathcal{S} \subset (\mathcal{V} \cup \mathcal{E})$. This simply involves concatenating samples from the vertex-domain with samples from the edge-domain, which can be expressed as 
\begin{equation*}
    \boldsymbol{x}_\mathcal{S} =\begin{bmatrix}
     \boldsymbol{\Psi}_\mathit{v}\\
     \boldsymbol{\Psi}_\mathit{e}
     \end{bmatrix}\boldsymbol{x}.
\end{equation*}
Denote the sampling budget by $M=M_1+M_2$, and the sampling matrix by $\boldsymbol{\Psi}=[\boldsymbol{\Psi}_v^\top,\boldsymbol{\Psi}_e^\top]^\top$, then the joint sampling in vertex-domain and edge-domain of a bandlimited graph signal can be rewritten as
\begin{equation*}
    \boldsymbol{x}_\mathcal{S} = \boldsymbol{\Psi}\boldsymbol{x} = \boldsymbol{\Psi}\boldsymbol{U}_B\boldsymbol{h}.
\end{equation*}
The signed sampling 
is to obtain the sign information by taking the sign of the sampled signal values, which can be described as
\begin{equation}
  \boldsymbol{y}=\text{sgn}\left(\boldsymbol{x}_\mathcal{S}\right)=\text{sgn} \left(\boldsymbol{\Psi}\boldsymbol{U}_B\boldsymbol{h}\right),
  \label{sign_sampling}
\end{equation}
where $\text{sgn}(\cdot)$ is
\begin{equation*}
  \text{sgn}(x)=
  \left\{
    \begin{array}{lr}
      -1      & x<0 \\
      1      & x>0 \\
      0                 & \text{otherwise}.
    \end{array}
  \right.
\end{equation*}
The signed sampling is a map from $\mathbb{R}^{B}$ to $\{-1,1,0\}^M$ according to (\ref{sign_sampling}). 

Due to the fact that the sign information does not depend on the magnitude of the signal, it only makes sense to recover the direction of $\boldsymbol{x}$\footnote{In practice, we often have prior knowledge about the magnitude of the signal, like in a goods rating system, where we know what the range of scores is. Sometimes we do not need to know the exact magnitude of the signal while knowing the ranking of the signal values on different vertices is sufficient.}. Hence, we make the \emph{unital assumption} that $\|\boldsymbol{x}\|_2=1$. 
Since the columns of $\boldsymbol{U}_B$ are orthonormal, we also have $\|\boldsymbol{h}\|_2=1$. In this paper, we are interested in recovering signal direction using sampled sign information. 

As illustrated in \figref{fig:1}, the central problem in this paper is to determine a sampling sequence where we take sign information, and subsequently recover the direction of the original bandlimited signal using the sign information acquired. This issue we are concerned with is abstracted from real-life scenarios, but few relevant methods are available.

\begin{figure}[h]
    \centering
    \includegraphics[scale=0.27]{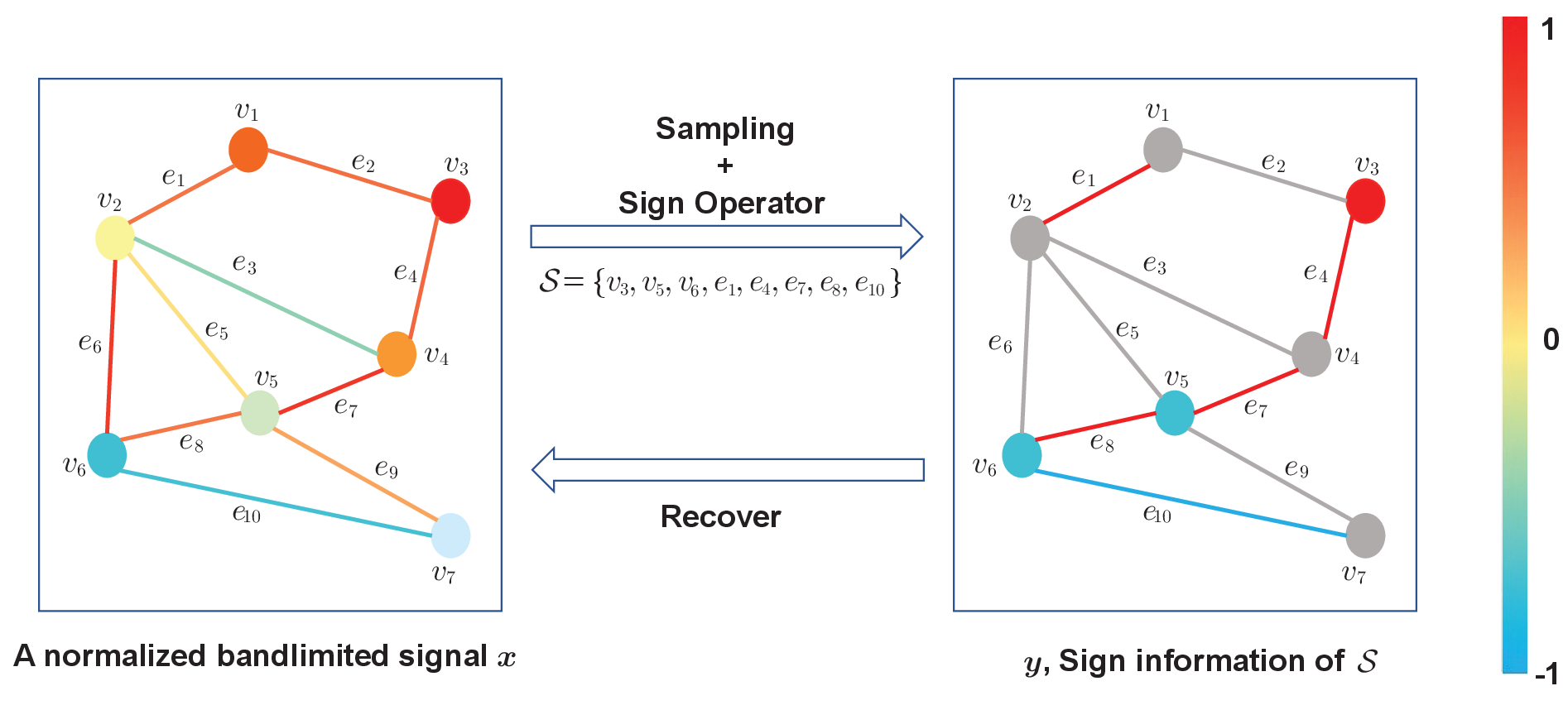}
    \caption{\label{fig:1} Sampling a signal on partial vertices and edges, the direction can be recovered from the sign information of the samples. In the example shown, the figure on the RHS only shows the signs of observed vertices and edges, which are used to recover the normalized signal on the LHS.}
\end{figure}

\subsection{Feasible Region}
\label{sec:feasible}

Denote the direction of the original graph signal by $\boldsymbol{x}^*=\boldsymbol{U}_B\boldsymbol{h}^*$, the sampling matrix by $\boldsymbol{\Psi}$, the set of samples by $\mathcal{S}$, the sign information by $\boldsymbol{y}$, and the passband of $\boldsymbol{x}^*$ by $\{f_1,f_2,\dots,f_B\}$. Since $\boldsymbol{U}_B$ is known, it is equivalent to recovering $\boldsymbol{h}^*$ instead of $\boldsymbol{x}^*$. 

Suppose our estimation of $\boldsymbol{h}$ is $\hat{\boldsymbol{h}}$. It is intuitive that $\hat{\boldsymbol{h}}$ should be consistent, i.e., $\text{sgn}(\boldsymbol{\Psi U}_B\hat{\boldsymbol{h}})=\boldsymbol{y}$. The condition is equivalent to the following constraints on $\hat{\boldsymbol{h}}$. For $i=1,2,\dots,M$, we have
\begin{equation}
    \label{C_v_1} 
    \begin{aligned}
        &\boldsymbol{\psi}_i^\top\boldsymbol{U}_B \hat{\boldsymbol{h}}< 0\quad \text{if}\ y_i < 0,\\
        &\boldsymbol{\psi}_i^\top\boldsymbol{U}_B \hat{\boldsymbol{h}}> 0\quad \text{if}\ y_i > 0,\\
        &\boldsymbol{\psi}_i^\top\boldsymbol{U}_B \hat{\boldsymbol{h}}= 0\quad \text{if}\ y_i = 0,
    \end{aligned}
\end{equation}
 where $\boldsymbol{\psi}_i^\top$ is the $i$-th row of $\boldsymbol{\Psi}$. 
 Define $\mathcal{C}$ as the intersection of sets described by $M$ constraints above, which is the finite intersection of open half-spaces and hyperplanes, then (\ref{C_v_1}) can be written as $\hat{\boldsymbol{h}}\in \mathcal{C}$.

Note that $\mathcal{C}$ is not closed, and we enlarge $\mathcal{C}$ to include the boundary points as 
\begin{equation}
    \label{sign_space_closed}
    \hat{\mathcal{C}}=\mathop{\bigcap}\limits_{\mathit{i}=1}^{\mathit{M}}\left\{\boldsymbol{w}\in\mathbb{R}^\mathit{B}\mid \,\boldsymbol{\psi }_\mathit{i}^\top\boldsymbol{U}_B\boldsymbol{w}\lesseqqgtr y_i\right\},
\end{equation}
where $x\lesseqqgtr a, a\in \mathbb{R}$ is defined as a constraint on $x$ depending on the value of $a$ as follows: 
\begin{equation*}
    \begin{aligned}
      \text{if}\ a<0, \text{ then $x$ must satisfy } x\leq 0; \\
      \text{if}\ a>0, \text{ then $x$ must satisfy } x\geq 0; \\
      \text{if}\ a =0, \text{ then $x$ must satisfy } x=0.
    \end{aligned}
\end{equation*}

The set $\hat{\mathcal{C}}$ is a closed convex cone in $\mathbb{R}^B$, as it is the finite intersection of closed half-spaces and hyperplanes. 
By the property of convex cones, when the number of half-spaces is sufficient enough, any vector in $\hat{\mathcal{C}}$ can be expressed as a linear combination of unit norm extreme vectors
with non-negative coefficients\cite{barker1973lattice}, i.e., we have 
\begin{equation}
    \label{EV}
    \hat{\mathcal{C}} = \left\{\sum_{i=1}^r k_i \boldsymbol{z}_i \mid \, k_i \geq 0, i=1,2,\dots,r \right\},
\end{equation}
where $\mathcal{Z}=\{\boldsymbol{z}_i\}_{i=1}^r$ is the set of extreme vectors (EVs) 
of unit norm in $\hat{\mathcal{C}}$. 

Recall that we have the extra condition $\|\boldsymbol{h}^* \|_2 = 1$ from the unital assumption. Hence, the  
 feasible region is the intersection of the convex cone $\hat{\mathcal{C}}$ and the unit sphere, expressed as
 \begin{equation*}
     \mathcal{J}=\hat{\mathcal{C}}\ \bigcap \  \{\boldsymbol{w} \in \mathbb{R}^B \mid \|\boldsymbol{w}\|_2=1 \}.
 \end{equation*}
To see an example with $B=3$, suppose we have already observed the sign information of $4$ samples.
By (\ref{sign_space_closed}), $\hat{\mathcal{C}}$ is obtained from $4$ constraints. Then the feasible region $\mathcal{J}$ is the intersection of $\hat{\mathcal{C}}$ and the unit sphere. \figref{fig:4}(a) gives an illustration of the feasible region in $\mathbb{R}^3$, where $\mathcal{J}$ (the gray region) is the intersection of $4$ halfspaces and the unit sphere. 

\emph{Key insight}: We want to leverage the intuition that the smaller the volume the feasible region $\mathcal{J}$ has, the smaller the possible value range of the signal is. Hence, we have a better chance to obtain an accurate recovery signal. Moreover, even if all of the sign information is observed, a perfect recovery is basically impossible. It is unlikely that the feasible region relying on sign information alone can shrink to a single point.

\section{The Problem of Online Sampling}
\label{sec:process}

In Section \ref{sec:model}, we propose the definition of obtaining the sign information of a graph signal with a given sampling set and discuss how to describe the corresponding feasible region. Through Section \ref{sec:model}, we convert the sign information of given samples into the conditions that the signal needs to meet.

In this section, we focus on how to solve the sampling problem. As mentioned earlier in Section~\ref{intro}, there are generally two ways to find a sampling set: offline and online. Offline sampling determines the whole sampling set altogether, which is based only on the objective function and prior information.
In other words, sampling in an offline manner cannot effectively utilize historical observations of the samples. Online sequential sampling can choose the next sample using the information gathered from the acquired samples. In contrast, online sampling takes into account more information and the selection of each sample is more comprehensive. 
Therefore, online sampling is more advantageous in terms of performance, and we try to determine the sampling set in an online manner.


\subsection{Sequential Evolution of the Feasible Region}
\label{problem_1}

Suppose there are $t$ observed samples $a_0, a_1, \dots, a_{t-1}$ with each $a_i$ either a vertex or an edge of $\mathcal{G}$. Suppose the associated sampling matrix is  $\boldsymbol{\Psi}=[\boldsymbol{\psi}_{a_0},\dots,\boldsymbol{\psi}_{a_{t-1}}]^\top\in \mathbb{R}^{t\times N}$ and the feasible region is $\mathcal{J}$. For a new sample $a_t$ with sign $y_t$, we may update $\boldsymbol{\Psi}$ by attaching a new row $\boldsymbol{\psi}_{a_t}^\top$. If $a_t$ is a vertex $v_k$, then $\boldsymbol{\psi}_{a_t}^\top$ is the $k$-th row of the identity matrix $\boldsymbol{I}_N$; while if $a_t$ is an edge $e_k$, then $\boldsymbol{\psi}_{a_t}^\top$ is the $k$-th row of the incidence matrix $\boldsymbol{\Xi}$. 
Correspondingly, the region $\mathcal{J}$ evolves as 
\begin{equation}
    \mathcal{J} \mapsto \mathcal{J}\bigcap \left\{\boldsymbol{w}\in\mathbb{R}^\mathit{B}\mid\,\boldsymbol{\psi}_{a_t}^\top\boldsymbol{U}_B\boldsymbol{w}\lesseqqgtr y_t\right\}.
    \label{next_hat_C}
\end{equation}
With each additional observed sample, the recovery signal needs to satisfy one more consistency constraint. Equivalently, selecting and observing a sample intersects a half-space or a hyperplane with the current feasible region, causing the feasible region to shrink once. A \emph{challenge} is that for each unsampled vertex or edge, we do not have access to its sign and hence the exact knowledge of how $\mathcal{J}$ evolves, before making the observation. 

We illustrate with a specific example shown in \figref{fig:4} the evolution of the feasible region for a new sample. 
\figref{fig:4}(a) and \figref{fig:4}(b) are the visualizations of the feasible regions before and after sampling $a_t$ respectively. \figref{fig:4}(c) is the top view of \figref{fig:4}(b). 
The gray region in \figref{fig:4}(a) represents the current feasible region $\mathcal{J}$, and the points A to D represent the EVs of $\hat{\mathcal{C}}$. The red line in \figref{fig:4}(b) and \figref{fig:4}(c) represents the hyperplane $\mathcal{H}=\{\boldsymbol{w}\in\mathbb{R}^\mathit{B}\mid\,\boldsymbol{\psi}_{a_t}^\top\boldsymbol{U}_B\boldsymbol{w}=0\}$. Suppose the part of $\mathcal{J}$ on the left side of $\mathcal{H}$ corresponds to $\mathcal{J}\cap \{\boldsymbol{w}\in\mathbb{R}^\mathit{B}\mid\,\boldsymbol{\psi}_{a_t}^\top\boldsymbol{U}_B\boldsymbol{w}\leq 0 \}$, while the other part corresponds to $\mathcal{J}\cap \{\boldsymbol{w}\in\mathbb{R}^\mathit{B}\mid\,\boldsymbol{\psi}_{a_t}^\top
\boldsymbol{U}_B\boldsymbol{w}\geq 0 \}$, denoted by $\mathcal{J}^1,\ \mathcal{J}^2$ respectively. We can see that $\mathcal{J}$ is separated by $\mathcal{H}$, and the next feasible region is uncertain, which could be $\mathcal{J}^1$, $\mathcal{J}^2$ or $\mathcal{J}\cap\mathcal{H}$, depending on the sign of $a_t$. For each unsampled vertex or edge, we have such a figure except that $\mathcal{J}$ is divided by different hyperplanes.

\begin{figure}[h]
    \centering
    \includegraphics[scale=0.27]{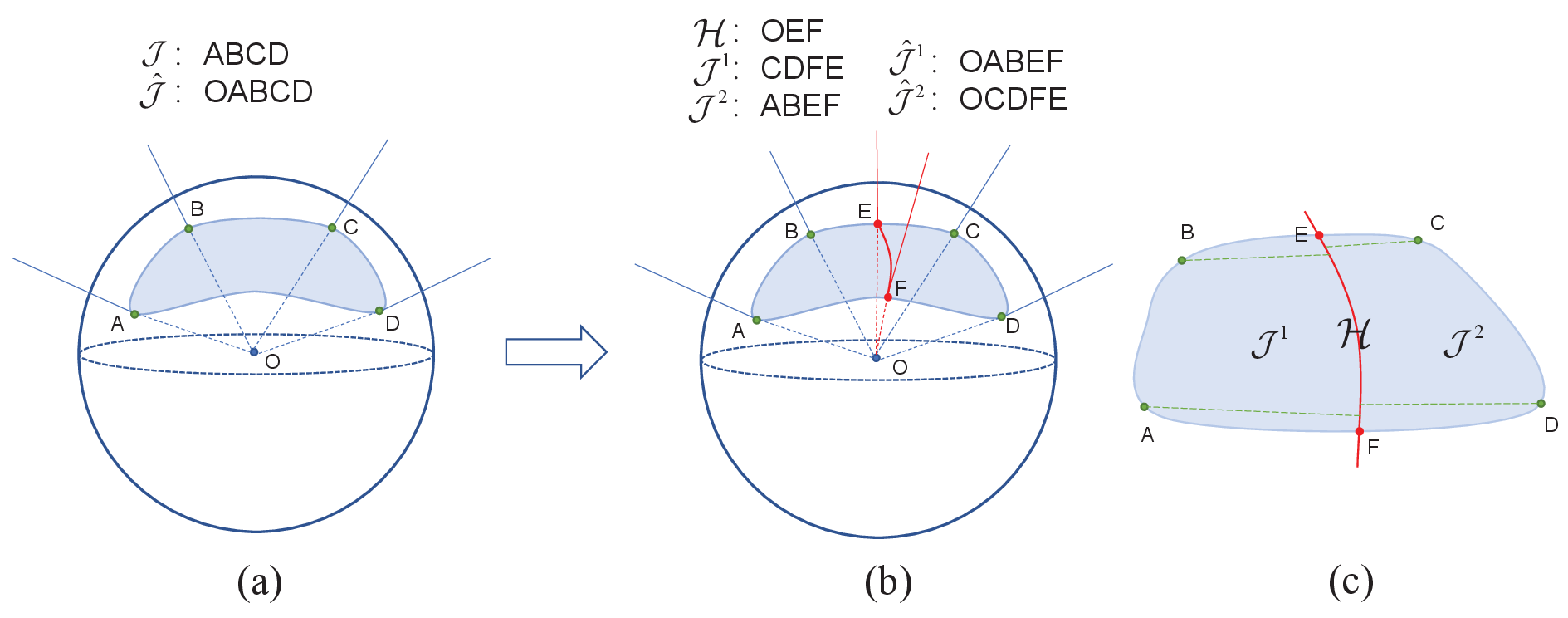}
    \caption{\label{fig:4} An example of $\mathcal{J}$ constructed by 4 halfspace constraints and the unital constraint with $a_t$ as the next sample in 3D space.} 
\end{figure}

As we have pointed out, we want to minimize the ``size'' of the feasible region. However, $\mathcal{J}$ has zero Euclidean volume. As an alternative, we define 
\begin{equation}
\label{body}
   \mathcal{B}_B = \{\boldsymbol{w} \in \mathbb{R}^B \mid \|\boldsymbol{w}\|_2 \leq 1 \},
\end{equation} 
and consider
\begin{equation}
\label{eq:hmh}
   \hat{\mathcal{J}} = \hat{\mathcal{C}}\ \bigcap \  \mathcal{B}_B.
\end{equation} 
It is the convex hull of $\mathcal{J}\cup \{\boldsymbol{0}\}$. There is a one-to-one correspondence between $\mathcal{J}$ and $\hat{\mathcal{J}}$. 
For example, in \figref{fig:4}, $\mathcal{J}$ is the face ABCD, while $\hat{\mathcal{J}}$ is the convex body OABCD, with $\mathcal{J}$ as the base. \emph{We use $\text{Vol}(\hat{\mathcal{J}})$, the volume of $\hat{\mathcal{J}}$ in $\mathbb{R}^B$, to measure the size of the feasible region.}

Significantly, although the observations are generated in an online manner, the feasible region according to (\ref{sign_space_closed}) does not depend on the time sequence. As long as the set of observation samples is the same, the resulting feasible region is the same, and the order of observations does not matter. This point will continue to be mentioned in later sections.

\subsection{The Online Sampling Decision Problem}
\label{problem_MDP}
The online sampling process can be modeled as an MDP and we use terminologies such as ``state'', and ``action'' from MDP theory \cite{bather2000decision}. In our model, the convex set $\hat{\mathcal{J}}$ at time step $t$ can be viewed as the state, constructed from the sampling sequence $a_{0:t-1}=\{a_0,a_1,\dots,a_{t-1}\}$ and the corresponding sign information $y_{0:t-1}=\{y_0,y_1,\dots,y_{t-1}\}$. Denote the observation sequence pair $(a_{0:t-1},y_{0:t-1})$ by $O_t$. We introduce the notation $\hat{\mathcal{J}}(O_t)$ for the state at time step $t$ or simply $\hat{\mathcal{J}}_t$ for convenience. 
Accordingly, denote the initial state as $\hat{\mathcal{J}}(\emptyset)$, or $\hat{\mathcal{J}}_0$.
The action space $\mathcal{A}_t$ is the set of all unsampled vertices and edges at time step $t$, i.e., the vertices and edges that are not included in $a_{0:t-1}$. 

As discussed before, once taking an action $a_t$ at state $\hat{\mathcal{J}}_t$, the feasible region would randomly transit to a new one $\hat{\mathcal{J}}_{t+1}$. 
Considering that any vector in $\mathcal{J}_t$ may be the original signal, we can naturally define that the transition probability of the feasible region is proportional to the size, i.e., 
\begin{equation*}
    \mathbb{P}_{a_t}(\hat{\mathcal{J}}_t,\hat{\mathcal{J}}_{t+1}) :=\mathbb{P}(\hat{\mathcal{J}}_{t+1} \mid \hat{\mathcal{J}}_t, a_t)=\frac{\text{Vol}(\hat{\mathcal{J}}_{t+1})}{\text{Vol}(\hat{\mathcal{J}}_{t})}.
\end{equation*}
For example, in \figref{fig:4}, the probability that the next feasible region is $\mathcal{J}^1$ is $\frac{\text{Vol}(\hat{\mathcal{J}}^1)}{\text{Vol}(\hat{\mathcal{J}})}$.  Note that 
$\mathcal{H}$ has $0$ Lebesgue measure in its ambient space, which means that it is almost impossible for the next feasible region to be $\mathcal{J}\cap\mathcal{H}$, compared to $\mathcal{J}^1$ and $\mathcal{J}^2$. Therefore, we ignore the case that the next feasible region is $\mathcal{J}\cap\mathcal{H}$ during the analysis. If $a_t$ is chosen as the next sample and the sign of $a_t$ is really observed to be $0$, the feasible region degenerates to a lower dimensional space, with the dimension subtracted by one relative to the original space. The subsequent sampling analysis remains unchanged and is carried out in this low-dimensional space.

After observing $y_t$, the state $\hat{\mathcal{J}}_{t+1}$ can be determined. Moreover, by (\ref{next_hat_C}), $\hat{\mathcal{J}}_{t+1}$ only depends on $\hat{\mathcal{J}}_t$ and $(a_t,y_t)$. 
The transition relationship can be expressed as
\begin{equation*}
    \hat{\mathcal{J}}_0\ \xrightarrow{a_0,y_0}\ \hat{\mathcal{J}}_1\  \xrightarrow{a_1,y_1}\  \dots\  \xrightarrow{a_{t-1},y_{t-1}}\  \hat{\mathcal{J}}_t
    \ \xrightarrow{a_t,y_t}\  \hat{\mathcal{J}}_{t+1}\ \dots
\end{equation*}

Besides, the reward $R_{a_t}(\hat{\mathcal{J}}_t,\hat{\mathcal{J}}_{t+1})$ of taking an action $a_t$ at state $\hat{\mathcal{J}}_t$ with next state $\hat{\mathcal{J}}_{t+1}$ can be defined as the volume reduction from $\hat{\mathcal{J}}_t$ to $\hat{\mathcal{J}}_{t+1}$, i.e., $R_{a_t}(\hat{\mathcal{J}}_t,\hat{\mathcal{J}}_{t+1})=\text{Vol}(\hat{\mathcal{J}_t})-\text{Vol}(\hat{\mathcal{J}}_{t+1})$. The reduction is uncertain during the sample selection stage as the sign is unobserved, but we can consider the expected reward. 
Note that $\hat{\mathcal{J}}_{t+1}$ can be written as $\hat{\mathcal{J}}(O_t\cup\{(a_t,y_t)\})$, which varies with $(a_t,y_t)$. 
The expected reward of $a_t$ at state $\hat{\mathcal{J}}_t$ can be written as
    \begin{align}
    \label{expected reward}
    &\mathbb{E}[R_{a_t}(\hat{\mathcal{J}}_t)]\nonumber\\=&\text{Vol}(\hat{\mathcal{J}}_t)-\sum_{y_t=+,-} \mathbb{P}(y_t\mid O_t)\text{Vol}\left(\hat{\mathcal{J}}(O_t\cup\{(a_t,y_t)\})\right), \nonumber 
    \end{align}
where

\begin{equation}
    \label{prob_yt}
    \mathbb{P}(y_t\mid O_t) = \frac{\text{Vol}\left(\hat{\mathcal{J}}(O_t\cup\{(a_t,y_t)\})\right)}{\text{Vol}(\hat{\mathcal{J}}_t)}.
\end{equation}

Taking \figref{fig:4} as an example, the expected reward of sampling $a_t$ is 
\begin{equation*}
    \mathbb{E}[R_{a_t}(\hat{\mathcal{J}})]=\text{Vol}(\hat{\mathcal{J}})-\frac{\text{Vol}(\hat{\mathcal{J}}^1)}{\text{Vol}(\hat{\mathcal{J}})}\text{Vol}(\hat{\mathcal{J}}^1)-\frac{\text{Vol}(\hat{\mathcal{J}}^2)}{\text{Vol}(\hat{\mathcal{J}})}\text{Vol}(\hat{\mathcal{J}}^2).
\end{equation*}
Based on such a setting, given a finite sampling budget, the sampling problem can be viewed as an online decision problem and analyzed under the framework of MDP \cite{bather2000decision}.


Assume that a total of $T$ decisions are made.
Similar to the general MDP problem, we try to maximize the expectation of the total reward produced by the online decisions, i.e., 
\begin{equation*}
    \begin{aligned}
        &\mathbb{E}\left[\sum_{t=0}^{T-1} R_{a_t}(\hat{\mathcal{J}}_t,\hat{\mathcal{J}}_{t+1})\mid \hat{\mathcal{J}}_0\right]\\ =& \mathbb{E}\left[\sum_{t=0}^{T-1} \left(\text{Vol}(\hat{\mathcal{J}}_t)-\text{Vol}(\hat{\mathcal{J}}_{t+1})\right)\mid \hat{\mathcal{J}}_0\right]\\
    =& \text{Vol}(\hat{\mathcal{J}}_0) - \mathbb{E}\left[\text{Vol}(\hat{\mathcal{J}}_T) \mid \hat{\mathcal{J}_0}\right].
    \end{aligned}
\end{equation*}
Therefore, 
if the whole sampling budget is $M$, then our sampling problem can be formulated as
\begin{equation}
    \begin{aligned}
    \label{optimization}
    \mathop{\text{max}}\limits_{a_{0:T-1}}\quad &\Phi(a_{0:T-1}) = \text{Vol}(\hat{\mathcal{J}}_0)-\mathbb{E}\left[\text{Vol}(\hat{\mathcal{J}}_T)\mid \hat{\mathcal{J}}_0 \right]\\
    \text{s.t.}\quad 
      &T \leq M,
  \end{aligned}
\end{equation}
where $a_t\in \mathcal{A}_t$ for each $t$, 
and $\hat{\mathcal{J}}_T$ is the final state, depending on the sign information of $a_{0:T-1}$. Symbol $\mathbb{E}[\cdot]$ is the conditional expectation over $(a_{0:T-1},y_{0:T-1})$, conditioned on the initial state.

Generally, an MDP problem in a finite horizon can be solved by dynamic programming (DP), based on the idea of backward induction. Typical methods include value iteration and policy iteration \cite{russell2010artificial}. Unfortunately, these methods are not suitable for our problem, due to the high complexity of the state space. For a graph with $N$ vertices and $|\mathcal{E}|$ edges, the sign information on each vertex and edge could be $+,-,0$ or not yet observed.
Therefore, with each combination corresponding to a state, there could be at $O(4^{N+|\mathcal{E}|})$ possible states. In the next section, we describe an approach that is more tractable. 

\section{Online Sampling and Recovery Algorithms}
\label{sec:Algorithm}
In this section, we turn to the problem of designing an online sampling policy such that the recovery result is as accurate as possible. Moreover, we provide the UPOCS algorithm to recover the original signal using the sign information of the selected samples. 

\subsection{The Sampling Policy and the GSS Algorithm}
\label{subsec:sampling}
 To be specific, we consider an online sampling scheme, which can effectively determine the new sample using the acquired sign information of the known samples.

Our policy is to choose the vertex or edge that makes the current feasible region be divided as evenly as possible. In other words, the feasible region is divided into two parts of as equal size as possible. The purpose is to optimize the worst-case scenario. To see an example (illustrated in \figref{fig:cut}), the face ABCD is the feasible region $\mathcal{J}$. Suppose $\mathcal{J}$ is divided into CDFE and ABEF by hyperplane OEF, and is divided into CDHG and ABGH by hyperplane OGH. From \figref{fig:cut}, OCDFE and OABEF have approximately equal volumes, while OCDHG and OABGH have a larger difference in volume. For the two options, we prefer \figref{fig:cuta}.

\begin{figure}[htbp]%
    \centering
    \subfloat[A sampling action that divides the feasible region into two parts with approximately equal volumes.]{
        \label{fig:cuta}
        \includegraphics[width=0.68\linewidth]{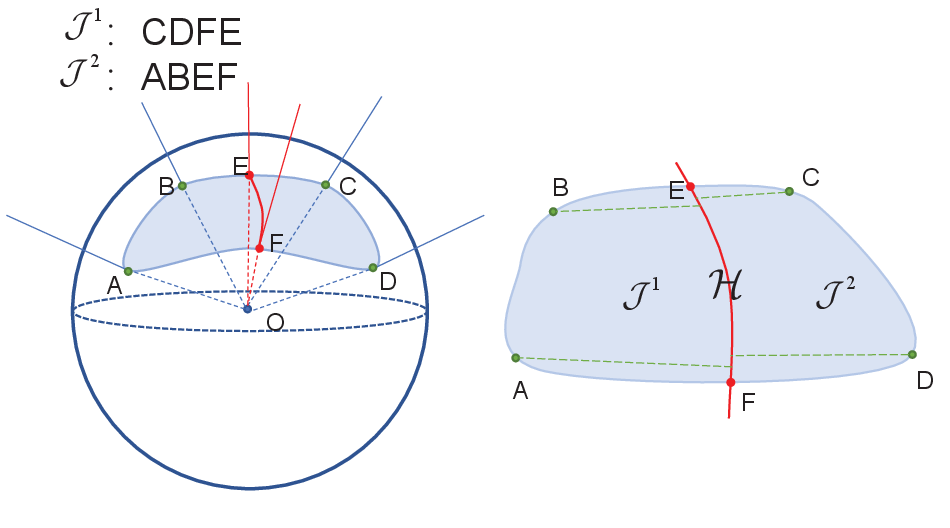}}
        \hfill
    \subfloat[A sampling action that divides the feasible region into two parts with a certain difference in volume.]{
        \label{fig:cutb}
        \includegraphics[width=0.68\linewidth]{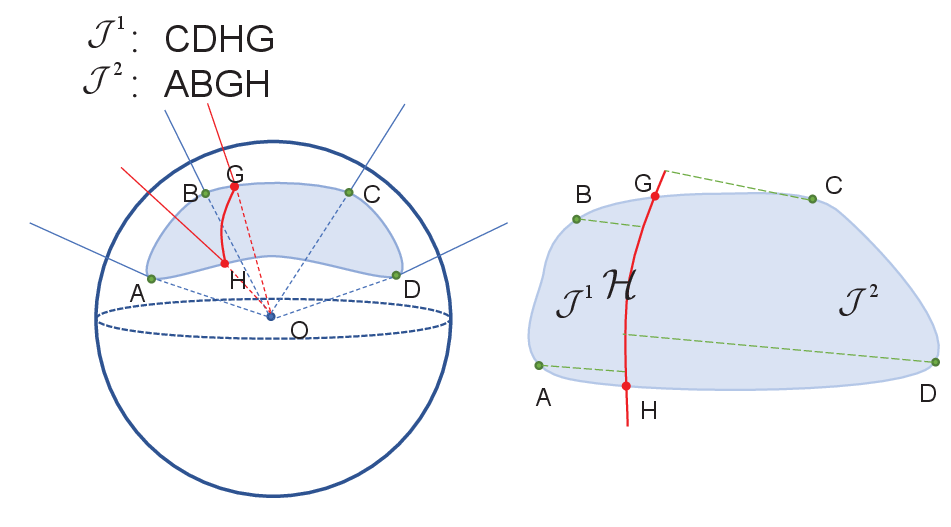}}
    \caption{ An example of the comparison of two sampling actions in 3D space.}
    \label{fig:cut}
\end{figure}

However, the volume of a convex body in a high-dimensional space is difficult to calculate exactly \cite{simonovits2003compute}, so it is necessary to apply efficient approximation methods.


We suggest an efficient approximation method, in the spirit of our previous work \cite{liu2022recovery}. Specifically, at time step $t$, consider the feasible region with extreme vectors (EVs) $\mathcal{Z}= \{\boldsymbol{z}_i\}_{i=1}^r$, and a candidate sample $a_t$ with the corresponding hyperplane $\mathcal{H}_{a_t}=\{\boldsymbol{w}\in \mathbb{R}^B \mid\ \boldsymbol{\psi}_{a_t}^\top\boldsymbol{U}_B\boldsymbol{w}=0\}$. We calculate the distance between each EV $\boldsymbol{z}$ and $\mathcal{H}_{a_t}$ as:
\begin{equation}
    \label{distance}
    d(\boldsymbol{z},\mathcal{H}_{a_t}) = \frac{\boldsymbol{\psi}_{a_t}^\top\boldsymbol{U}_B\boldsymbol{z}}{\|\boldsymbol{\psi}_{a_t}^\top\boldsymbol{U}_B\|}.
\end{equation} 
The sample is determined by
\begin{equation}
    \label{ICASSP}
    a_t^* =\mathop{\text{argmin}}\limits_{a_t \in \mathcal{A}_t}\ \left|\sum_{\boldsymbol{z} \in \mathcal{Z}} d(\boldsymbol{z},\mathcal{H}_{a_t})\right|.
\end{equation}
We explain the intuition using the example in \figref{fig:4}. The length of the dotted line in \figref{fig:4}(c) is the distance from each EV to $\mathcal{H}_{a_t}$, which can be calculated by (\ref{distance}). Note that A, B, and C, D lie on different sides of $\mathcal{H}_{a_t}$, so the sign of $d(\boldsymbol{z},\mathcal{H}_{a_t})$ for A, B and for C, D would be different. Summing them up amounts to the difference in total distance between EVs from both sides and the hyperplane. Therefore, we use (\ref{ICASSP}) 
as an approximation of the difference between the volumes of the regions on the two sides of the hyperplane. The vertex or edge with a smaller difference is preferred. 

We want to use the greedy algorithm based on the above policy. To start the process, we also need to determine the initial state $\hat{\mathcal{J}}_0$. The challenge is that there is not any EV for less than $B-1$ samples. For example, it is stated in \cite[Definition 4.2]{bertsimas1997introduction} that an extreme ray of a polyhedral cone in $\mathbb{R}^B$ is generated from $B-1$ linearly independent inequalities. Therefore, we need to initialize by finding the initial samples. 

For this, we choose the indices of the linearly independent rows with the maximum magnitude of $[\boldsymbol{U}_B^\top,(\boldsymbol{\Xi U}_B)^\top]^\top$, to form the initial sampling set. To be specific, we sort the magnitudes of the row vectors of $[\boldsymbol{U}_B^\top,(\boldsymbol{\Xi U}_B)^\top]^\top$ in descending order, select the first $B-1$ linearly independent rows, and take the corresponding vertices and edges as the initial samples. After observing the $B-1$ initial samples, only one EV can be determined.

Although we are not yet able to use our approximation method, it is possible in the selection of the $B$-th sample to produce the effect of cutting the feasible region roughly in half. 
As for the $B$-th sample, 
we search all the remaining vertices and edges and choose the vertex or edge such that the maximum angle between EVs is minimized, because such a hyperplane would be approximately in the middle of the feasible region. In this way, the feasible region can be approximately equally cut, which results that the worst effect of the $B$-th sample is not too bad. 
It needs to be acknowledged that this is an intuitively valid method of initialization, and other reasonable initialization methods are worth trying.

Once the initial $B$ samples are observed, the initial feasible region can be determined, and our greedy sampling process (\ref{ICASSP}) can be applied. Denote the set of initial samples by $\mathcal{S}_0$, 
then the action space $\mathcal{A}_B$ of our sampling decision process at the initial state is $(\mathcal{V}\cup\mathcal{E})\backslash\mathcal{S}_0$, including all the remaining unsampled vertices and edges.



\begin{algorithm}
	\renewcommand{\algorithmicrequire}{\textbf{Input:}}
	\renewcommand{\algorithmicensure}{\textbf{Output:}}
  \renewcommand{\algorithmicprint}{\textbf{Initialize:}}
	\caption{Greedy Signed Sampling (GSS).} 
	\label{alg:2} 
	\begin{algorithmic}[1]
		\REQUIRE Sampling budget $M$, passband $\{f_1,f_2,\dots,f_B\}$, graph $\mathcal{G}$
		\ENSURE Sampling sequence $\mathcal{S}$
    \STATE Get a list by sorting the magnitudes of row vectors of $[\boldsymbol{U}_B^\top,(\boldsymbol{\Xi U}_B)^\top]^\top$ in descending order.
    \STATE Find the first $B-1$ indices in the list to construct $\mathcal{S}_0$ such that the corresponding vectors are of full rank.
    \STATE Observe $\mathcal{S}_0$, $t \gets B-1$
    \FOR {$i$ not in $\mathcal{S}_0$}
        \STATE Calculate EVs $\mathcal{Z}_1,\mathcal{Z}_2$ for $y_i =+,-$. 
        \STATE $\theta_i = \text{min} (\text{min}_{\boldsymbol{z}_1,\boldsymbol{z}_2 \in \mathcal{Z}_1}\langle \boldsymbol{z}_1,\boldsymbol{z}_2\rangle,\text{min}_{\boldsymbol{z}_1,\boldsymbol{z}_2 \in \mathcal{Z}_2}\langle \boldsymbol{z}_1,\boldsymbol{z}_2 \rangle)$
    \ENDFOR
    \STATE $a_t^* = \text{argmax}_i\ \theta_i$, $\mathcal{S}_0 \gets \mathcal{S}_0 \cup \{a_t^*\}$
    \STATE Observe $a_t^*$, and calculate EVs $\mathcal{Z}$. 
    \STATE $\mathcal{S} \gets \mathcal{S}_0$, $t \gets B$
    \WHILE{$t < M$}
    \FOR {$a_t$ not in $\mathcal{S}$}
    \STATE Calculate $d(\boldsymbol{z},\mathcal{H}_{a_t})$ for each $\boldsymbol{z}\in\mathcal{Z}$ using (\ref{distance}).
    \ENDFOR
    \IF{the stopping criterion (Section~\ref{stopping criterion}) is met}
    \STATE \textbf{return} $\mathcal{S}$
    \ENDIF
    \STATE Find $a_t^*$ as in (\ref{ICASSP}).
    \STATE $\mathcal{S} \gets \mathcal{S} \cup \{a_t^*\}$, $t \gets t+1$
    \STATE Observe $a_t^*$, and update EVs $\mathcal{Z}$.
	\ENDWHILE 
    \STATE \textbf{return} $\mathcal{S}$
	\end{algorithmic} 
\end{algorithm}

In summary, the proposed GSS algorithm is presented in Algorithm \ref{alg:2}. Steps 1 to 8 are the initialization, and the remaining steps form the online sampling process. 
To run the GSS algorithm, make sure the following conditions are met.
\begin{enumerate}
    \item The GSS algorithm should be applied on an undirected and connected graph, with nonnegative edge weights.
    \item The target signal should be bandlimited, whose nonzero expansion coefficients under some given orthonormal matrix are concentrated in a passband.
    \item The target signal should also be unital, because the sign information of the output samples from the GSS algorithm cannot be used to infer the amplitude of the original signal and can only be used to estimate the direction.
    \item Generally, the sampling budget should be set to: $B < M < |\mathcal{V}| + |\mathcal{E}|$. 
\end{enumerate}
We shall further justify the GSS algorithm in Section~\ref{ssec:greedy}, and the superior performance of the GSS algorithm in sampling sign information over existing methods would be confirmed in Section~\ref{sec:experiments}. 

\subsection{Signal Recovery Given Sign Information}
\label{sec:recovery}
In this subsection, we propose the signal recovery algorithm UPOCS, given the acquired sign information of samples. Note that the feasible region is the intersection of a convex cone and the unit sphere. The idea is that we first get a solution that satisfies the observation constraints, and then we normalize it for the recovery signal.

For a convex cone $\hat{\mathcal{C}}$ described by (\ref{sign_space_closed}), we have the projection operators onto $\hat{\mathcal{C}}$ \cite{theodoridis2010adaptive}. Suppose the $i$-th sample in $\mathcal{S}$ is vertex $v_j$, then the projection onto $\hat{\mathcal{C}}$ can be defined as
\begin{equation}
    \boldsymbol{P}\boldsymbol{w}:=\left\{ \begin{array}{lr}
    \boldsymbol{w} - \frac{(\boldsymbol{U}_B)_j^\top(\boldsymbol{U}_B)_j}{\|(\boldsymbol{U}_B)_j\|^2} \boldsymbol{w},     \\\qquad\qquad\qquad\quad\  \text{if}\ v_j\in \mathcal{S},\ \text{sgn} \big((\boldsymbol{U}_B)_j\boldsymbol{w}\big)
      \not=y_\mathit{i}; \\
      \boldsymbol{w},  \qquad\qquad\qquad\qquad\qquad\qquad\qquad\quad   \text{otherwise}.
    \end{array}
    \right.
  \label{Pv}
\end{equation}
where $(\boldsymbol{U}_B)_j$ is the $j$-th row of $\boldsymbol{U}_B$.
Suppose the $i$-th sample in $\mathcal{S}$ is edge $e_j = (v_p,v_q)$ and $p < q$, then the projection onto $\hat{\mathcal{C}}$ can be defined as
\begin{equation}
    \boldsymbol{P}\boldsymbol{w}:=\left\{ \begin{array}{llr}
    \boldsymbol{w} - \frac{[(\boldsymbol{U}_B)_p-(\boldsymbol{U}_B)_q]^\top[(\boldsymbol{U}_B)_p-(\boldsymbol{U}_B)_q]}{\|(\boldsymbol{U}_B)_p-(\boldsymbol{U}_B)_q\|^2} \boldsymbol{w}, \\ \quad\quad\ \text{if}\ v_j\in \mathcal{S},
       \text{sgn} \big((\boldsymbol{U}_B)_p\boldsymbol{w}-(\boldsymbol{U}_B)_q)\boldsymbol{w}\big)\not=y_\mathit{i}; \\
      \boldsymbol{w},     \qquad\qquad\qquad\qquad\qquad\qquad\qquad\quad             \text{otherwise}.
    \end{array}
    \right.
  \label{Pe}
\end{equation}
Intuitively, for any signal $\boldsymbol{w}$ and vertex $v_j \in \mathcal{S}$, if the sign of $(\boldsymbol{U}_B\boldsymbol{w})_j$ is not consistent with the given sign information (outside the constraint space), then project $\boldsymbol{w}$ to the hyperplane $\{\boldsymbol{w}\in\mathbb{R}^\mathit{B}\mid\,(\boldsymbol{U}_B)_j\boldsymbol{w}= 0\}$, which is the boundary of the constraint space; otherwise, it is unchanged. Similarly, for any edge $e_j =(v_p,v_q) \in \mathcal{S}$, if the sign of $(\boldsymbol{U}_B\boldsymbol{w})_p -(\boldsymbol{U}_B\boldsymbol{w})_q$ is not consistent with the given sign information, then project $\boldsymbol{w}$ to $\{\boldsymbol{w}\in\mathbb{R}^\mathit{B}\mid\,(\boldsymbol{U}_B)_p\boldsymbol{w}-(\boldsymbol{U}_B)_q\boldsymbol{w}= 0\}$, otherwise, it keeps unchanged. 
As can be seen later in Appendix \ref{A-A}, the defined projection operator is firmly non-expansive.

Based on the definitions of $\hat{\mathcal{C}}$, along with the projection operator $\boldsymbol{P}$, and inspired by the classical projections onto convex sets algorithm (POCS) \cite{bauschke1996projection}, we propose that the direction of the original signal can be estimated by iteratively projecting onto $\hat{\mathcal{C}}$.  
Specifically, starting from an arbitrary signal, then a simple iterative projection process is applied. With enough iterations, the output is close enough to $\hat{\mathcal{C}}$. 

After the projections, we perform a normalization, then the recovery signal is obtained. Formally, the UPOCS algorithm can be described in Algorithm \ref{alg:1}.

\begin{algorithm} 
	\renewcommand{\algorithmicrequire}{\textbf{Input:}}
	\renewcommand{\algorithmicensure}{\textbf{Output:}}
	\caption{Unital Projection Onto Convex Sets (UPOCS).} 
	\label{alg:1} 
	\begin{algorithmic}[1]
		\REQUIRE Samples $\mathcal{S}$, an arbitrary random signal $\boldsymbol{h}_0$, maximum iterations $n_\text{max}$
		\ENSURE a recovery signal $\hat{\boldsymbol{h}}$ 
            \STATE Observe $\mathcal{S}$, and construct the projection operator $\boldsymbol{P}$.  
		\STATE $n \gets 0$ 
		\WHILE{$n < n_\text{max}$} 
		\STATE $\boldsymbol{h}_{\mathit{n}+1}\gets\boldsymbol{P}\boldsymbol{h}_\mathit{n}$ 
		\STATE $n \gets n+1$
		\ENDWHILE 
		\STATE \textbf{return} $\boldsymbol{h}_\mathit{n} / \|\boldsymbol{h}_n\|$ 
	\end{algorithmic} 
\end{algorithm}

\section{Analysis}
\label{theory}
\subsection{Analysis of the GSS Algorithm}
\label{ssec:greedy}

As mentioned in Section~\ref{sec:process}, signed sampling is a decision process for a Markovian model. Starting from an initial state, we decide on a sample with the aim of minimizing the size of the feasible region in the long run, then the sign of the sample is acquired and the next state is observed. This process is repeated until the sampling budget is reached.

Consider an example illustrated by \figref{fig:5}. Suppose the volume of the current state $\hat{\mathcal{J}}_t$ is $50$. If $\hat{\mathcal{J}}_t$ is divided into two parts by a sample $a_t$, whose volumes are $20$ and $30$ respectively, as the transition probability is proportional to the volume, there are $40\%$ and $60\%$ chances that the two parts would be the next state $\hat{\mathcal{J}}_{t+1}$ respectively. Such a feasible region segmentation and state transition process are repeated at time step $t+1, t+2$. 
\figref{fig:5} shows various possible state transitions of the sampling sequence $\{a_t,a_{t+1},a_{t+2}\}$, where the numbers in the boxes represent the volume of the state at each time step; while the numbers on the arrows represent the corresponding state transition probabilities. 

Note that the last branch of the tree in \figref{fig:5} has probabilities $0$ and $1$, which means the sign of $a_{t+2}$ is certain and the next state is no longer random. It can also be explained that the current feasible region would not be divided by the  hyperplane in \figref{fig:4}. Therefore, in this case, $a_{t+2}$ is not selected, because in the worst case, the feasible region does not have a volume reduction for the next state.

\begin{figure}[h]
    \centering
    \includegraphics[scale=0.25]{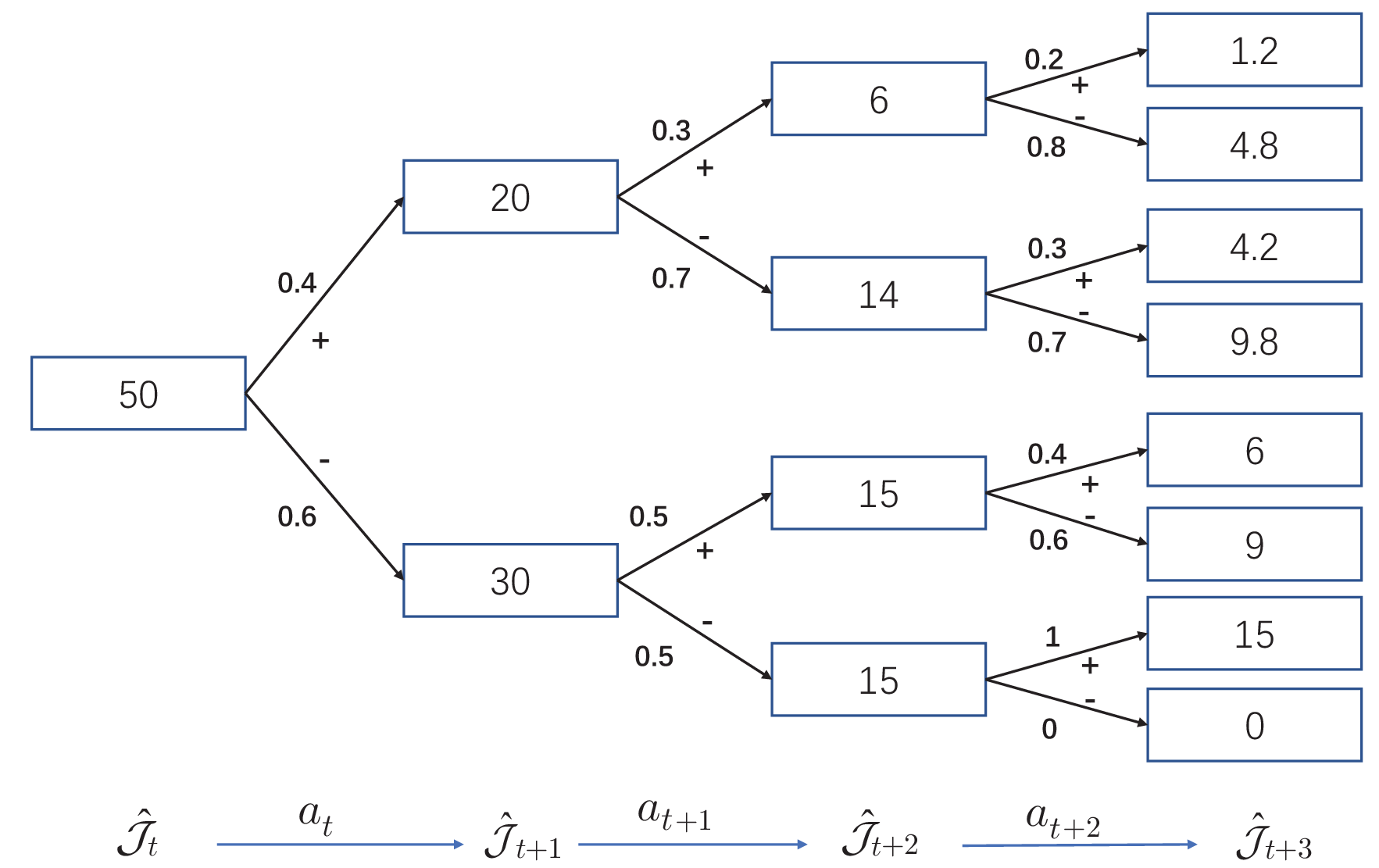}
    \caption{\label{fig:5} An example of feasible region segmentation for a sampling sequence.}
\end{figure}

We call the structure in \figref{fig:5} the feasible region segmentation tree. For each different sampling sequence, we can draw a tree like this, and each complete path (from the root node to some bottom leaf node) stands for a sign combination of the sampling sequence. In such a binary tree, each node has a number associated with the volume of a region. The sum of the numbers of the children nodes equals that of the parent node. The structure of the tree has important properties, which we now study.

Consider the possible sign information of all elements in $\mathcal{A}_0=\mathcal{V}\cup\mathcal{E}$, which is a random variable. Let $\mathcal{Q}$ be a \textit{complete realization} of $\mathcal{A}_0$. 
Each $\mathcal{Q}$ corresponds to a set of (vertex/edge, sign) pairs, that specifies the sign information of all elements in $\mathcal{A}_0$. 
A \textit{partial realization} is a subset of some $\mathcal{Q}$, which can be viewed as a set of (vertex/edge, sign) pairs, specifying the sign information of partial elements in $\mathcal{A}_0$ (the elements we have already picked from $\mathcal{A}_0$). 
For example, the observation sequence $O_t$
can be viewed as a partial realization. 
A partial realization $O_t$ is \textit{consistent} with a complete realization $\mathcal{Q}$ if they are equal everywhere in $a_{0:t-1}$. 
If $O_{t_1}=(a_{0:t_1-1},y_{0:t_1-1})$ and $O_{t_2}=(a_{0:t_2-1},y_{0:t_2-1})$ ($t_1\leq t_2$) are both consistent with a complete realization $\mathcal{Q}$, and each element in $a_{0:t_1-1}$ is also in $a_{0:t_2-1}$, 
then $O_{t_1}$ is a \textit{subrealization} of $O_{t_2}$. 

We illustrate with an example. Suppose for a given graph signal, we have $\mathcal{A}_0 = \{a_0,a_1,\dots,a_4\}$, where $a_i\ (i=0,1,\dots,4)$ is an unsampled vertex or edge. 
Assume the sign information of $a_0,a_1,\dots,a_4$ is $+,+,-,-,+$  respectively.
We select and observe $a_0,a_1,a_2$ sequentially, and the observations are $\{(a_0,+),(a_1,+),(a_2,-)\}$. 
In this example, any $\{(a_0,y_0),\dots,(a_4,y_4)\}(y_i \in \{+,-,0\}, i=0,1,\dots,4)$ is a complete realization, which describes the sign information of each element in $\mathcal{A}_0$. Among them, $\{(a_0,+),(a_1,+),(a_2,-),(a_3,-),(a_4,+)\}$ is the complete realization $\mathcal{Q}$ that is consistent with the true sign information of the signal.
$\{(a_0,+),(a_1,+),(a_2,-)\}$ is a partial realization, 
which is consistent with $\mathcal{Q}$. If we select $a_3$ at the next time step and observe its sign, then $\{(a_0,+),(a_1,+),(a_2,-),(a_3,-)\}$ is also a partial realization, which has a subrealization $\{(a_0,+),(a_1,+),(a_2,-)\}$. 
Although the observations are generated in an online manner, we consider two observation sequences to be the same partial realization as long as the samples in both observation sequences are the same and the sign information is consistent. For example, observation in the order of $a_0,a_1,a_2$ produces the same outcome as that in the order of $a_2,a_0,a_1$.

Let $\mathcal{O}$ be the collection of all the partial realizations and complete realizations. We introduce the following definitions.

\begin{definition}[\textbf{Monotonicity}]
A function $\pi:\mathcal{O}\rightarrow \mathbb{R}$
is monotone if and only if $\pi(O_{t_1}) \leq \pi(O_{t_2})$ for all $O_{t_1},O_{t_2}\in \mathcal{O}$, where $O_{t_1}$ is a subrealization of $O_{t_2}$. 
\end{definition}

For a function $\pi:\mathcal{O}\rightarrow \mathbb{R}_{\geq 0}$ 
and a given partial realization $O_t$,
the \textit{conditional expected marginal benefit} for an element $a_t$ conditioned on having observed $O_t$ is
\begin{equation*}
    \Delta(a_t\mid O_t) := \mathbb{E}\left[\pi(O_t\cup \{(a_t,y_t)\}) 
    -\pi(O_t)\mid O_t \right],
\end{equation*}
where $y_t$ stands for the sign of $a_t$, and the expectation is taken over $y_t$ \cite{golovin2011adaptive}. Using this, we further introduce the following.

 
\begin{definition}[\textbf{Adaptive Monotonicity} \cite{golovin2011adaptive}]
 A function $\pi:\mathcal{O} \rightarrow \mathbb{R}_{\geq 0}$ is adaptive monotone if the conditional expected marginal benefit of any element is nonnegative, i.e., for all $O_t=(a_{0:t-1},y_{0:t-1})$ and $a_t \in \mathcal{A}_t$, we have
$\Delta(a_t\mid O_t)\geq 0$.
\end{definition}

\begin{definition}[\textbf{Adaptive Submodularity} \cite{golovin2011adaptive}]
 A function $\pi:\mathcal{O} \rightarrow \mathbb{R}_{\geq 0}$ is adaptive submodular if for all $O_{t_1}=(a_{0:t_1-1},y_{0:t_1-1})$, $O_{t_2}=(a_{0:t_2-1},y_{0:t_2-1})$ ($O_{t_1}$ is a subrealization of $O_{t_2}$), and all $a \in \mathcal{A}_{t_2}$, we have $\Delta(a\mid
O_{t_1})\geq \Delta(a\mid O_{t_2})$.
\end{definition}


Back to our online sampling problem, let 
$\hat{\mathcal{J}}_0$ be the corresponding initial state. 
Based on the listed definitions, we have the following findings.
\begin{lemma}
 \label{branch}
    Let $\varphi(O_t)=-\text{Vol}(\hat{\mathcal{J}}(O_t))$, where 
    $\hat{\mathcal{J}}(O_t)$ is the state described by $O_t$. 
    Then $\varphi$ is monotone.
\end{lemma}
\begin{proof}
    See Appendix \ref{A-B} for details.
\end{proof}

It implies that each complete path of the feasible region segmentation tree (cf.\ \figref{fig:5}) satisfies monotonicity.
\begin{lemma}
 \label{submodular}
  Let $\phi(O_t)=\text{Vol}(\hat{\mathcal{J}}_0)-\text{Vol}(\hat{\mathcal{J}}(O_t))$,
  where $\hat{\mathcal{J}}(O_t)$ is the state described by $O_t$. 
  Then $\phi$ is adaptive monotonic and adaptive submodular.
\end{lemma}
\begin{proof}
    See Appendix \ref{A-C} for details.
\end{proof}

 As mentioned in Section~\ref{sec:Algorithm}, our policy is to greedily select the vertex or edge so that the current feasible region can be divided into two parts with the smallest difference in size. The policy can be expressed as
\begin{align}
    \label{equal volume}
        a_t^* = \mathop{\text{argmin}}\limits_{a_t\in \mathcal{A}_t}\ \Big\vert &\text{Vol}\left(\hat{\mathcal{J}}(O_t^*\cup \{(a_t,+)\})\right)\nonumber\\
        & - \text{Vol}\left(\hat{\mathcal{J}}(O_t^*\cup \{(a_t,-)\})\right)\Big\vert, 
\end{align}
where $a_t^*$ is the sample to be determined by the above greedy method at time step $t$. 
$\mathcal{A}_t$ is the set of all unsampled vertices and edges at time step $t$, i.e., $\mathcal{A}_t=\mathcal{A}_0\backslash a_{0:t-1}^*$. 
Moreover, $O_t^*$ is the observation sequence associated with $a_{0:t-1}^*$, and $\hat{\mathcal{J}}(O_t^* \cup \{(a_t,y_t)\})$ 
with $y_t=+,-$ is the state at time step $t+1$.

\begin{theorem}
    \label{main conclusion}
    Suppose the sampling sequence solved from (\ref{equal volume}) is $a_{0:T-1}^*$, and the globally optimal sampling set of (\ref{optimization}) with the same size is $a_{0:T-1}^\text{opt}$. Then we have
    \begin{equation*}
        \Phi(a_{0:T-1}^*)  \geq \left(1-\frac{1}{e}\right)\Phi(a_{0:T-1}^\text{opt}).
    \end{equation*}
\end{theorem}
\begin{proof}
    See Appendix \ref{A-D} for details.
\end{proof}

Therefore, the sampling sequence solved in an online manner by (\ref{equal volume}) can achieve the performance within $(1-\frac{1}{e})$ of the globally optimal solution of (\ref{optimization}). This justifies the online sampling process (Steps 11 to 21) of the GSS algorithm, because
(\ref{ICASSP}) approximates the implementation of (\ref{equal volume}) using the difference in the spatial distance instead of the difference in volume.

\subsection{Stopping Criterion of the GSS Algorithm}
\label{stopping criterion}

This subsection discusses an early stopping criterion for the GSS algorithm. Note that the online sampling process is motivated by reducing the size of the feasible region. If we search all the unsampled vertices and edges and find that the feasible region would not change anymore, then we can terminate the sampling process because there will be no more reward. From another perspective, the sign information on all the unsampled vertices and edges can be inferred with certainty, and further sampling only adds redundant samples. 

\begin{figure}[h]
    \centering
    \includegraphics[scale=0.26]{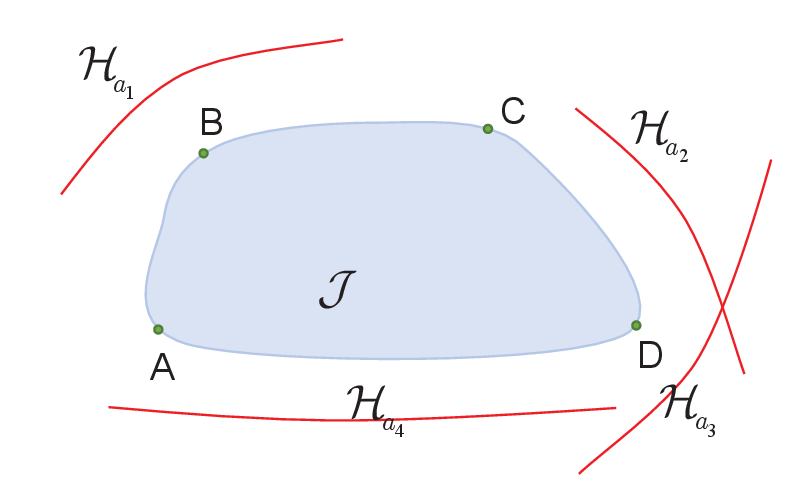}
    \caption{\label{fig:stop} An example of feasible region segmentation during the sampling process that can be terminated pre-maturely.}
\end{figure}

To see an example, in \figref{fig:stop}, the gray region depicts the feasible region $\mathcal{J}$ (top view), and the red lines $\mathcal{H}_{a_1},\ldots,\mathcal{H}_{a_4}$ are the hyperplanes associated with the remaining unsampled vertices and edges $a_1,\ldots,a_4$. As we can see, $\mathcal{H}_{a_1},\ldots, \mathcal{H}_{a_4}$ make no contribution to the reduction of the size of the feasible region because $\mathcal{J}$ (and $\hat{\mathcal{J}}$) would not be separated. 
In this case, the feasible region segmentation tree (cf.\ \figref{fig:5}) degenerates into a single chain, and the feasible region remains unchanged for any subsequent samples. In general, we have the following result.

\begin{theorem}
 \label{the:stopping}
  Let $O_t=(a_{0:t-1},y_{0:t-1})$ be the historical observations, and $\mathcal{J}_t$, $\hat{\mathcal{J}}_t$ be the respective feasible region and state at time step $t$. If at some time step ${t_1}$, for any $a_{t_1} \in \mathcal{A}_{t_1}$, $\text{Vol}(\hat{\mathcal{J}}_{t_1})=\text{Vol}(\hat{\mathcal{J}}_{{t_1}+1})$ holds, 
  then $\text{Vol}(\hat{\mathcal{J}}_{t_2}) = \text{Vol}(\hat{\mathcal{J}}_{t_1})$ for any $O_{t_2}$, where $O_{t_1}$ is a subrealization of $O_{t_2}$. 
\end{theorem}
In other words, if any unsampled vertex or edge does not decrease the size of the feasible region, then the size of the feasible region will remain the same no matter how many additional samples are added.

\begin{proof}
    Suppose at time step ${t_1}$, the remaining candidate vertices and edges are $a_{j_1},\dots,a_{j_m}$  
    with the corresponding sign information $y_{a_{j_1}},\dots y_{a_{j_m}}$.  
    If any unsampled vertex or edge does not cause the size of the feasible region to decrease, referring to (\ref{next_hat_C}), we have 
    \begin{align}
    \label{pre_no_intersection}
            \mathcal{J}_{t_1} \bigcap \left\{\boldsymbol{w}\in\mathbb{R}^\mathit{B}\mid\boldsymbol{\psi}_{a_{j_i}}^\top\boldsymbol{U}_B\boldsymbol{w}\lesseqqgtr y_{a_{j_i}}\right\} = \mathcal{J}_{t_1}, \nonumber\\
            \text{for all}\ i\in \{1,2,\dots,m\}.
    \end{align}
    
     Through this condition, $\mathcal{J}_{t_1}$ is in the intersection of the constraint spaces of the remaining unsampled vertices and edges. So, the feasible region remains unchanged no matter how many the samples in $\{a_{j_1},\dots,a_{j_m}\}$ are added.
    \end{proof}

From the result, we can stop sampling if (\ref{pre_no_intersection}) holds. Without loss of generality, suppose that some candidate sample $a_{j_i}$ has sign $y_{a_{j_i}} = -1$, then (\ref{pre_no_intersection}) can be rewritten as
\begin{equation*}
    \mathcal{J}_{t_1} \bigcap \left\{\boldsymbol{w}\in\mathbb{R}^\mathit{B}\mid\boldsymbol{\psi}_{a_{j_i}}^\top\boldsymbol{U}_B\boldsymbol{w}\leq 0\right\} = \mathcal{J}_{t_1},
\end{equation*}
i.e., 
\begin{equation*}
    \mathcal{J}_{t_1} \bigcap \left\{\boldsymbol{w}\in\mathbb{R}^\mathit{B}\mid\boldsymbol{\psi}_{a_{j_i}}^\top\boldsymbol{U}_B\boldsymbol{w}> 0\right\} = \emptyset.
\end{equation*}
In other words, the current feasible region $\mathcal{J}_{t_1}$ is located entirely on one side of the hyperplane $\left\{\boldsymbol{w}\in\mathbb{R}^\mathit{B}\mid\boldsymbol{\psi}_{a_{j_i}}^\top\boldsymbol{U}_B\boldsymbol{w} = 0\right\}$.

Therefore, it is time to terminate the sampling once
the current feasible region is on one side of each hyperplane of the remaining unsampled vertices and edges. 
As each vector in the feasible region is a linear combination of EVs
with non-negative coefficients, if all the EVs are on one side of each hyperplane associated with the remaining unsampled vertices and edges, so does the entire feasible region. 
The stopping criterion is thus given as follows.
\newline
   \newline \textbf{Stopping criterion.}     \textit{In the greedy selection stage of the GSS algorithm, 
    if at some time step $t$,  $d(\boldsymbol{z},\mathcal{H}_{a_t}) \geq 0$ for each $\boldsymbol{z} \in \mathcal{Z}$ and $a_t \in \mathcal{A}_t$, or $d(\boldsymbol{z},\mathcal{H}_{a_t})\leq 0$, for each $\boldsymbol{z} \in \mathcal{Z}$ and $a_t \in \mathcal{A}_t$,
    then the sampling process is terminated.}
    \newline

Using the stopping criterion, the GSS algorithm can be terminated pre-maturely (even the sample size is smaller than the budget $M$); and by then, all the remaining unsampled vertices and edges do not provide additional knowledge. Importantly, 
once the stopping criterion is reached, the sampling sequence solved by the GSS algorithm can theoretically result in the same recovery performance as the case that all the vertices and edges are observed. 

\subsection{Performance Analysis of the UPOCS Algorithm}
In this subsection, we discuss a theoretical result associated with the UPOCS algorithm. 
The following theorem about the convergence properties of the UPOCS algorithm is derived, 
through the study of the feasible region and the projection operator.
\begin{theorem}
    \label{theorem1}
  The iterative recovery sequence $\left\{\boldsymbol{h}_\mathit{n}\right\}$ converges linearly to some point $\hat{\boldsymbol{h}}$ in $\hat{\mathcal{C}}$, and the convergence
  rate is independent of the choice of the initial point $\boldsymbol{h}_0$.
\end{theorem}

\begin{proof}
  See Appendix \ref{A-A} for details.  
\end{proof}

It is worth noting that the recovery result depends on $\boldsymbol{h}_0$, though the convergence rate is independent of $\boldsymbol{h}_0$.

\section{Experiments}
\label{sec:experiments}

In this section, we use numerical experiments to demonstrate the effectiveness of our approach. 

\subsection{Recovery of Synthesis Graph Signals}
\label{sensor graph}

In this experiment, we generate three graph topologies, with the following parameters
\begin{enumerate}
    \item Sensor graph: $N = 40, |\mathcal{E}| = 138$.
    \item Erdős–Rényi (ER) graph: $N = 40$, and the connection probability $p = 0.3$, which results $|\mathcal{E}| = 98$.
    \item Watts-Strogatz (WS) graph: $N = 40$, average degree $k = 4$, and the reconnection probability $p = 0.25$, which results $|\mathcal{E}| = 160$.
\end{enumerate}
The graph topologies are shown in Fig.\ref{topology}.

\begin{figure}[htbp]%
    \centering
    \subfloat[]{
        \label{topology:sensor}
        \includegraphics[width=0.45\linewidth]{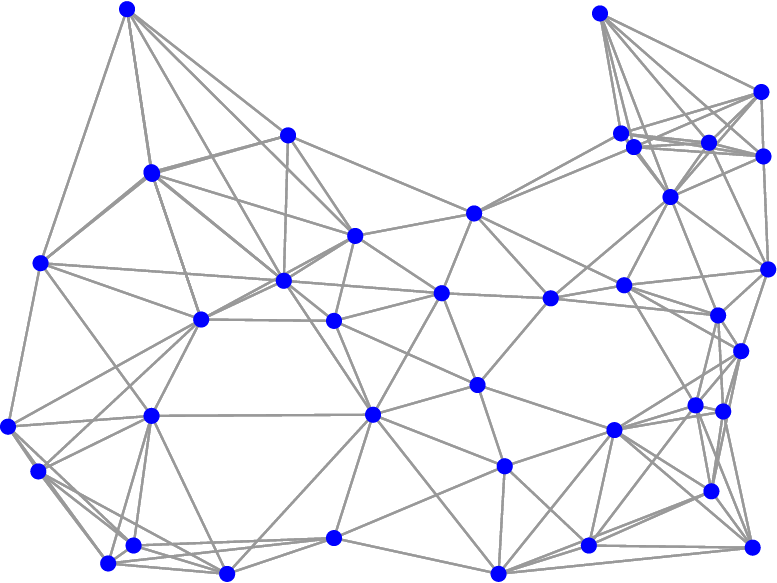}}
        \hfill
    \subfloat[]{
        \label{topology:ER}
        \includegraphics[width=0.45\linewidth]{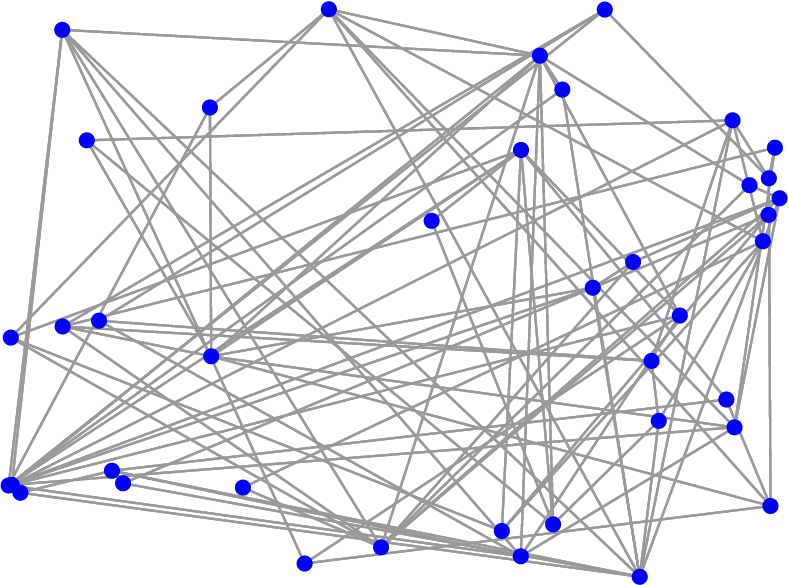}}
        \hfill
    \subfloat[]{
        \label{topology:SW}
        \includegraphics[width=0.45\linewidth]{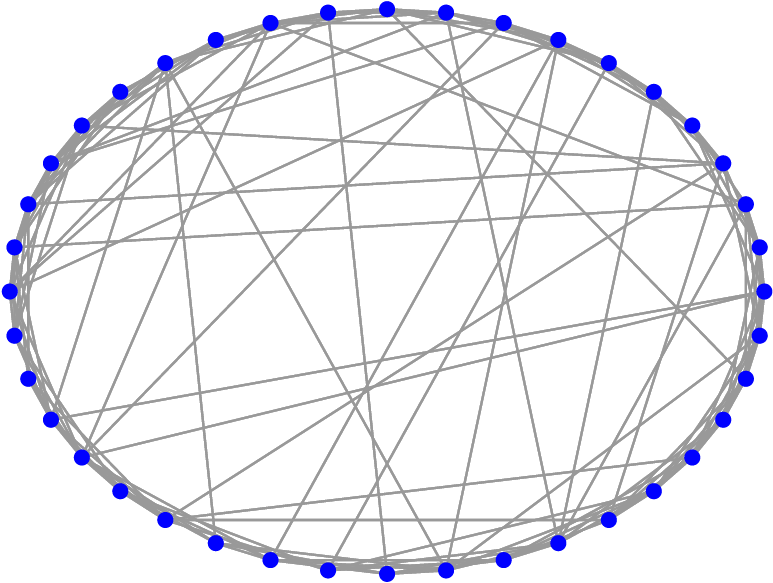}}
    \caption{The topologies of the (a) sensor graph, (b) ER graph, and (c) WS graph. 
    }
    \label{topology}
\end{figure}

On these graph topologies, set $B = 7$ and $f_1, f_2, \dots, f_B = 29, 30, \dots, 35$. The frequency coefficients are obtained from a 0-1 uniform distribution, and the corresponding bandlimited signals with unit norm are generated. Next, we examine the recovery quality of different sampling methods.

We first obtain the sampling sequences by the GSS algorithm and a few benchmarks, then acquire the corresponding sign information. 
After this, we select $K$ signals as the initial signals to recover the direction of the original signal by the UPOCS algorithm and compare the recovery performance of each sampling algorithm using the average error with the following metric
\begin{equation}
    \label{error in angle}
    \delta = \frac{1}{K}\mathop{\sum}\limits_{i=1}^K \arccos\langle \boldsymbol{x}^*,\hat{\boldsymbol{x}}_i \rangle.
\end{equation}
$\hat{\boldsymbol{x}}_i$ is the normalized recovery signal using the UPOCS algorithm for the $i$-th initial signal, and $\boldsymbol{x}^*$ is the direction of the original signal, which is also normalized. 
The $\arccos$ function measures the angle between the input vectors of unit length. The recovery has poor quality if $\delta$ is large.

We set $K=50$ and the number of iterations in the UPOCS algorithm is $10^4$. We first consider the case that only vertices are sampled. The results measured by $\delta$ are presented in \figref{fig:6_v}. 

\begin{figure}[htbp]%
    \centering
    \subfloat[]{
        \label{fig6:sensor}
        \includegraphics[width=0.75\linewidth]{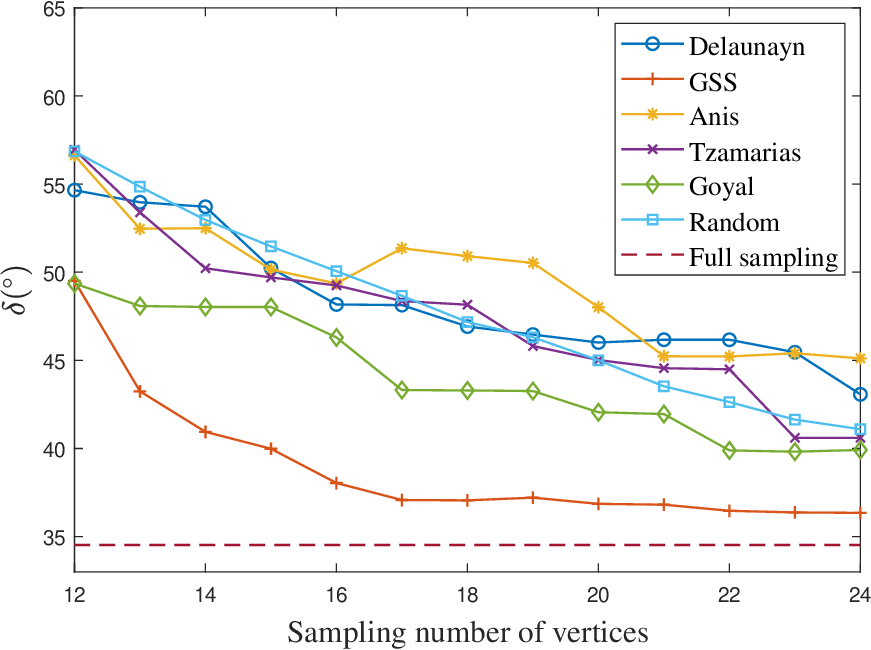}}
        \hfill
    \subfloat[]{
        \label{fig6:ER}
        \includegraphics[width=0.75\linewidth]{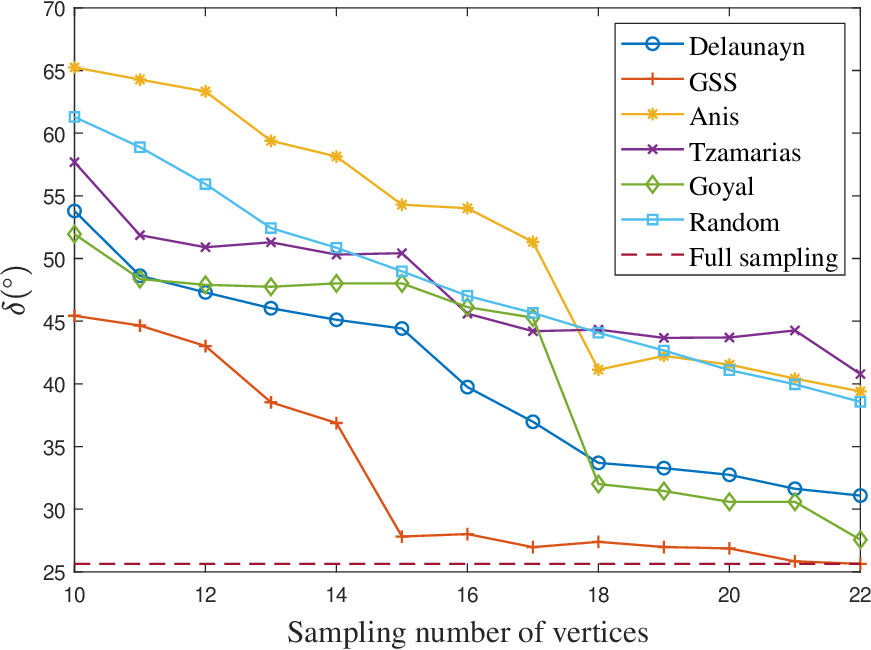}}
        \hfill
    \subfloat[]{
        \label{fig6:SW}
        \includegraphics[width=0.75\linewidth]{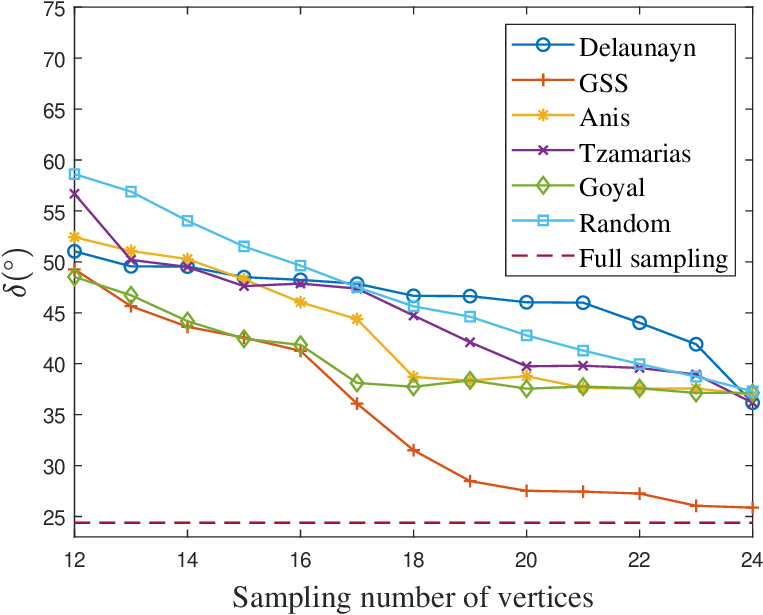}}
    \caption{Comparison of recovery performance between sampling algorithms for different sampling numbers when only vertices can be sampled on the (a) sensor graph, (b) ER graph, and (c) WS graph. 
    }
    \label{fig:6_v}
\end{figure}

In \figref{fig:6_v}, ``Delaunayn" refers to the method of simplex subdivision to estimate the volume in (\ref{equal volume}) \cite{henk2017basic}. We use the 'Delaunayn' function in MATLAB. 
``GSS" is our greedy sampling method presented in Algorithm \ref{alg:2}, where the EVs are calculated through pycddlib\footnote{Link: \!\url{https://pycddlib.readthedocs.io/en/latest/}}. 
``Anis" is the sampling framework of \cite{anis2014towards} for continuous graph signals. ``Tzamarias" is the method in \cite{tzamarias2018novel}, which is a novel sampling algorithm for continuous signals based on the concept of uniqueness set.
``Goyal" is the sampling method in \cite{goyal2018estimation} of selecting the indices corresponding to the row vectors in 
$\boldsymbol{U}_B$ with the largest lengths. 
Moreover, we averaged the recovery error over $50$ random sampling sets, plotted in \figref{fig:6_v} with the label ``Random". In addition, we also compare these methods with full sampling, which can be viewed as the ideal case that all the vertices are observed.  

\begin{figure}[h]%
    \centering
    \subfloat[]{
        \label{fig7:sensor}
        \includegraphics[width=0.7\linewidth]{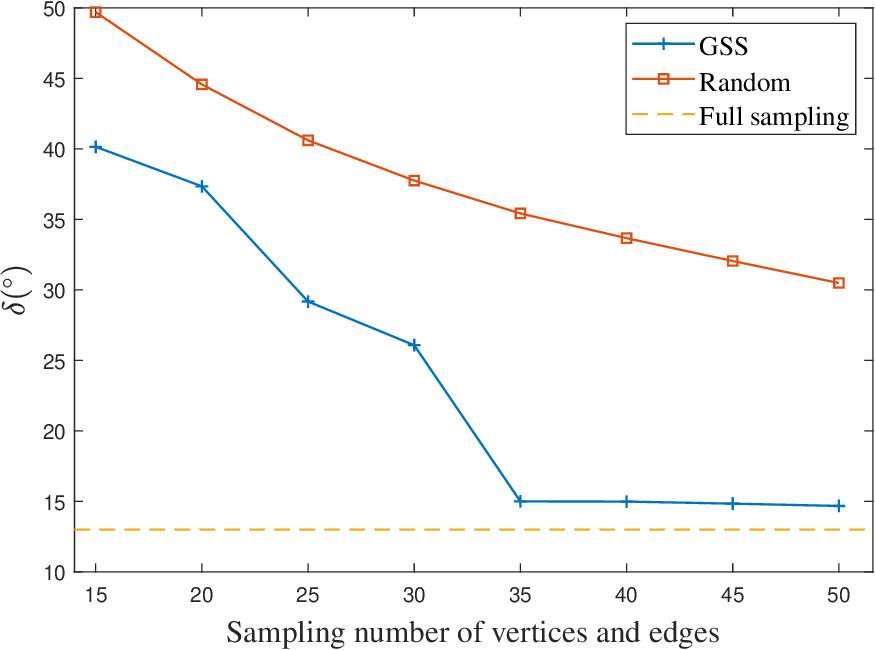}}
        \hfill
    \subfloat[]{
        \label{fig7:ER}
        \includegraphics[width=0.7\linewidth]{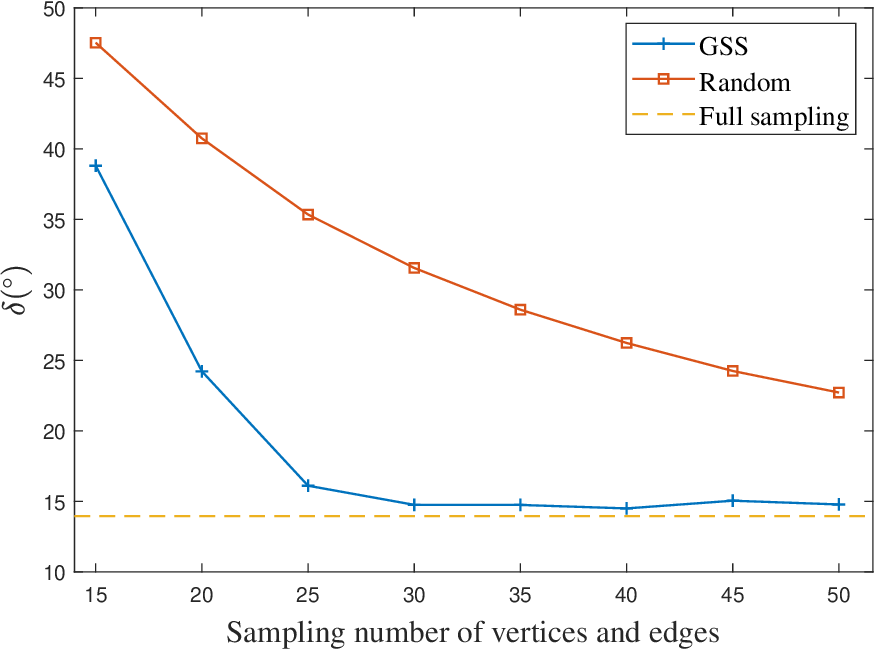}}
        \hfill
    \subfloat[]{
        \label{fig7:SW}
        \includegraphics[width=0.7\linewidth]{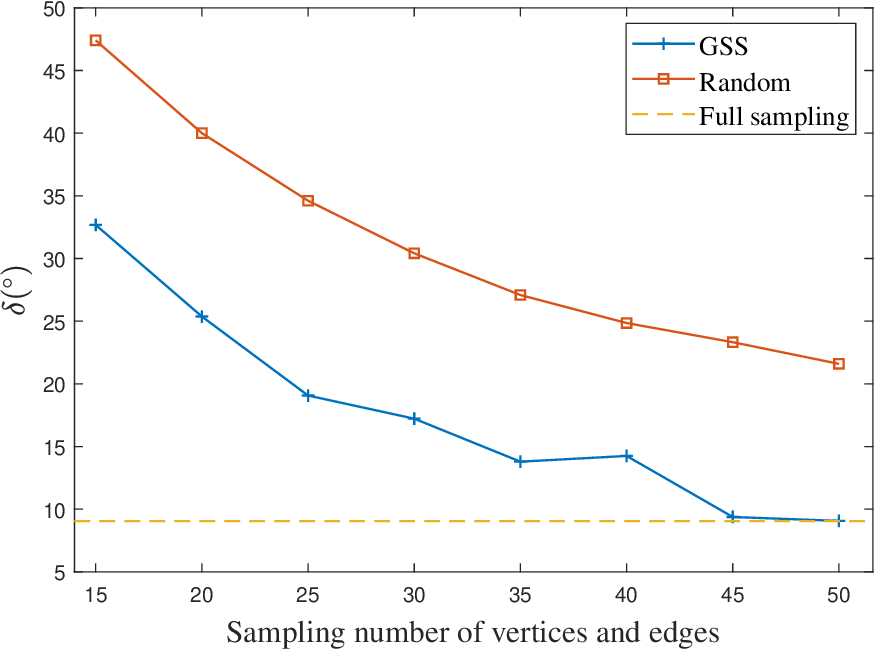}}
    \caption{Comparison of $\delta$ between sampling algorithms for different sampling numbers when both vertices and edges can be sampled on the (a) sensor graph, (b) ER graph, and (c) WS graph. 
    }
    \label{fig:7_ve}
\end{figure}

As shown in \figref{fig:6_v}, the GSS algorithm outperforms most benchmarks as the metric $\delta$ of the GSS algorithm is the lowest compared with other benchmarks, except ``Full sampling''. The result is very close to that of ``Full sampling'' even if only nearly half of the vertices are sampled and observed. 
We notice that the recovery performance of ``Anis" and ``Tzamarias" are even worse than random sampling, which indicates that these methods may not be suitable for signed sampling. Both ``Anis" and ``Tzamarias" consider the sampling of continuous values, focusing on the exact magnitude of the signal, while for signed sampling, the focus is on whether the signal is above a threshold rather than the exact magnitude. 

On the other hand, we consider the case that both vertices and edges can be sampled, which means there are $|\mathcal{V}|+|\mathcal{E}|$ samples to be selected. \figref{fig:7_ve} shows the average error in angle $\delta$ of the GSS algorithm, random sampling, and the ideal case that all the vertices and edges are sampled and observed. Similar to \figref{fig:6_v}, the GSS algorithm has a better performance than random sampling. Besides, with less than half of the total sample size, the GSS algorithm achieves a very close result to that of the ideal case (``Full sampling'').

\begin{figure}[htbp]%
    \centering
    \subfloat[]{
        \label{fig8:sensor}
        \includegraphics[width=0.75\linewidth]{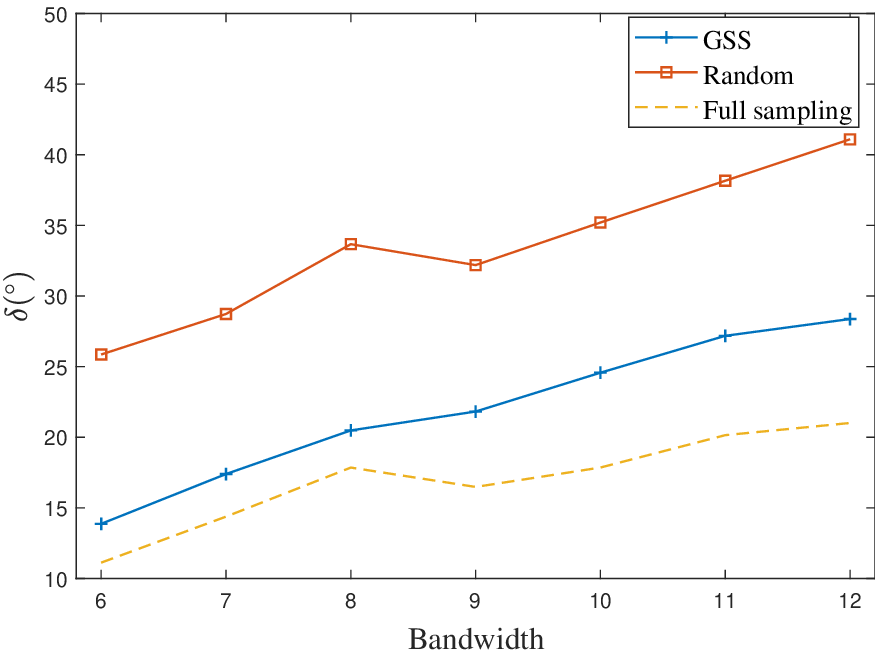}}
        \hfill
    \subfloat[]{
        \label{fig8:ER}
        \includegraphics[width=0.75\linewidth]{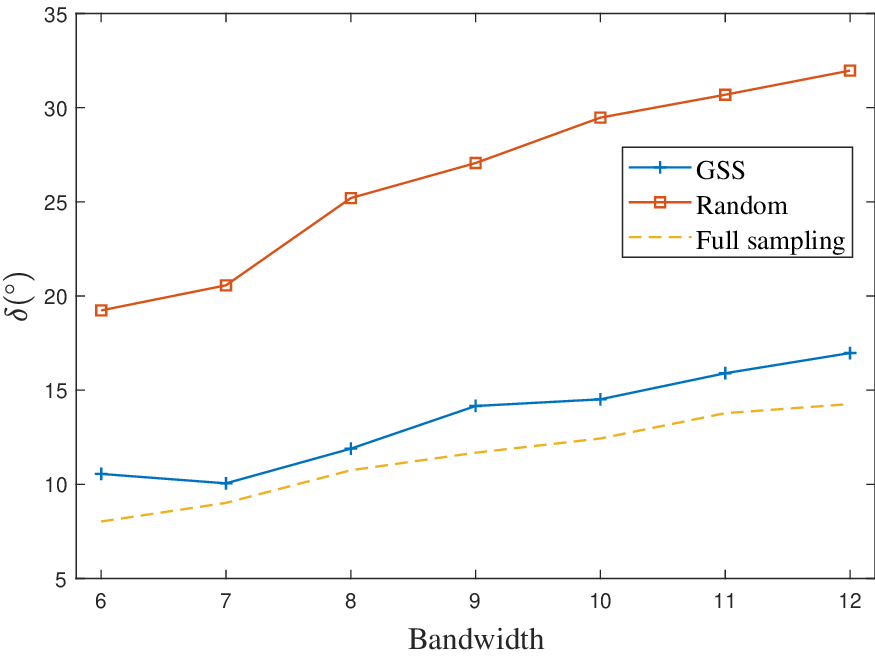}}
        \hfill
    \subfloat[]{
        \label{fig8:SW}
        \includegraphics[width=0.75\linewidth]{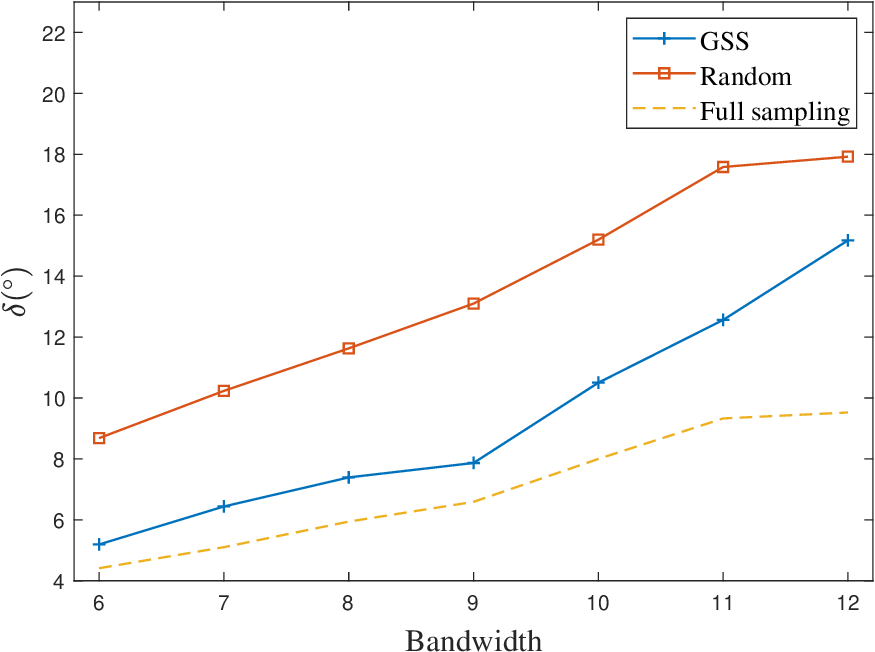}}
    \caption{Comparison of $\delta$ between sampling algorithms for different bandwidths and fixed sampling budget when both vertices and edges can be sampled on the (a) sensor graph, (b) ER graph, and (c) WS graph. 
    }
    \label{fig8_b}
\end{figure}

In addition, another experiment is conducted to explore the effect of bandwidth on recovery. We fix the sampling budget to $50$, and set several bandwidth values. For each bandwidth,  we generate $50$ signals on each graph and then recover them. The average results are presented in \figref{fig8_b}. It can be concluded that the larger the bandwidth of a signal, the more difficult it is to achieve satisfactory recovery with a given sampling budget.

To gain more intuition, the effect of bandwidth on the recovery can be further analyzed. Notice that the bandwidth determines the spatial dimension in which the sampling and recovery algorithms work.
Assuming that the expected recovery error (error in angle) does not exceed $\delta$, then the feasible region needs to be narrowed down to within the surface of the unit spherical cap with colatitude angle $\delta$ in $\mathbb{R}^B$. 
As can be found in \cite{li2010concise}, on the unit sphere, the area of a spherical cap in $B$-dimensional Euclidean space is
\begin{equation*}
    A_B^\text{cap} = \frac{1}{2}A_B I_{\sin^2\delta}\left(\frac{B-1}{2}, \frac{1}{2}\right),
\end{equation*}
where $A_B$ is the area of the unit sphere, and $I_x(a, b)$ is the regularized incomplete beta function. Intuitively, for bandwidth $B$, the difficulty of having the recovery error within $\delta$ is of the order $O\left(I_{\sin^2\delta}^{-1}\left(\frac{B-1}{2}, \frac{1}{2}\right)\right)$. By numerical experiments, $I_x(\frac{B-1}{2}, \frac{1}{2})$ is decreasing with respect to $B$ for a certain range of integers at $0 < x < 1$, as shown in Fig.~\ref{ratio}. Therefore, as the bandwidth increases, the ratio of the ideal feasible region becomes smaller, which requires more samples (constraints). This is consistent with our experimental results in Fig.~\ref{fig8_b}.

\begin{figure}[h]
    \centering
    \includegraphics[width=0.7\linewidth]{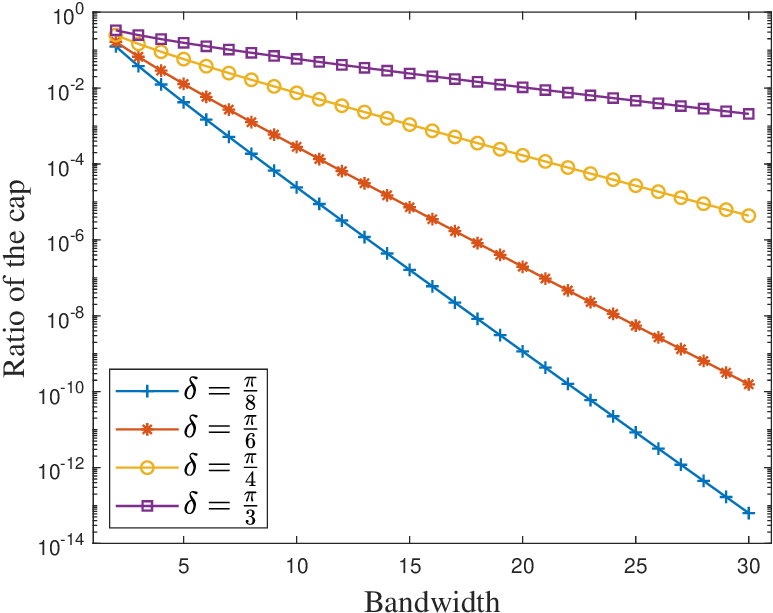}
    \caption{\label{ratio} The ratio between the area of the hyperspherical cap and that of the unit sphere under different bandwidths.}
\end{figure}

\subsection{Recovery and Classification of Realworld Data}
\label{realdata}
In this experiment, we consider the recovery and classification tasks on a real dataset. Specifically, we collect a dataset of ratings for movies and TV shows\footnote{Link:\!\url{https://www.kaggle.com/code/jyotmakadiya/popular-movies-and-tv-shows-data-analysis/data}} from Kaggle, where the ratings are numbers between $0$ and $10$ with one decimal place. 

\begin{figure}[h]
    \centering
    \includegraphics[width=0.7\linewidth]{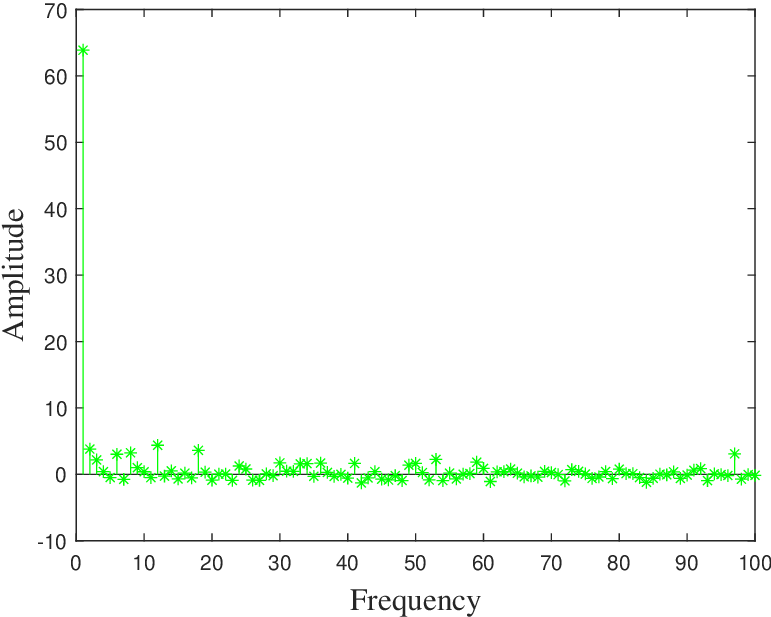}
    \caption{\label{spectral} Signal spectrum of the rating dataset in \ref{realdata}.}
\end{figure}

We construct a graph with $100$ vertices and $4596$ edges where each vertex corresponds to a movie or a TV show. The edge weights are generated according to the similarity of attributes of the vertices, such as the release date, and ages of the audience. We take these attributes as components of attribute vectors, and then compute the inner product of these vectors for each vertex pair as edge weights. 
To facilitate calculation and analysis, edges are sparsified appropriately and outliers are removed to ensure connectivity. The ratings are regarded as the original graph signal, whose (GSP) spectrum is shown in \figref{spectral}, which is concentrated in the low frequency region. We round up the signal value of each vertex as the label (ground truth), e.g., a vertex with a score of $4.8$ is labeled $5$.

Assume that the score range is known. Considering that the original signal is only approximately bandlimited and is positive on each vertex, to acquire the sign information, we carry out the following processing.
\newline (1) Remove the DC component (the component of $0$ frequency).
\newline (2) Choose frequency components with the largest amplitude. More specifically, we apply a bandpass filter to the signal, preserving the frequency components with the largest amplitudes. By choosing $B=13$, the energy after filtering accounts for about $75\%$ of the original energy. 

The estimation of the original signal is obtained as follows.
\newline (1) Perform signed sampling of vertices and edges on the resulting signal after the above pre-processing and recover it using the UPOCS algorithm ($K=30$ and $n_\text{max}=3000$). 
\newline (2) The DC component is added to the recovery signal.
\newline (3) Scale to the known score range.

\begin{figure}[htbp]%
    \centering
    \subfloat[]{
        \label{fig:10a}
        \includegraphics[width=0.7\linewidth]{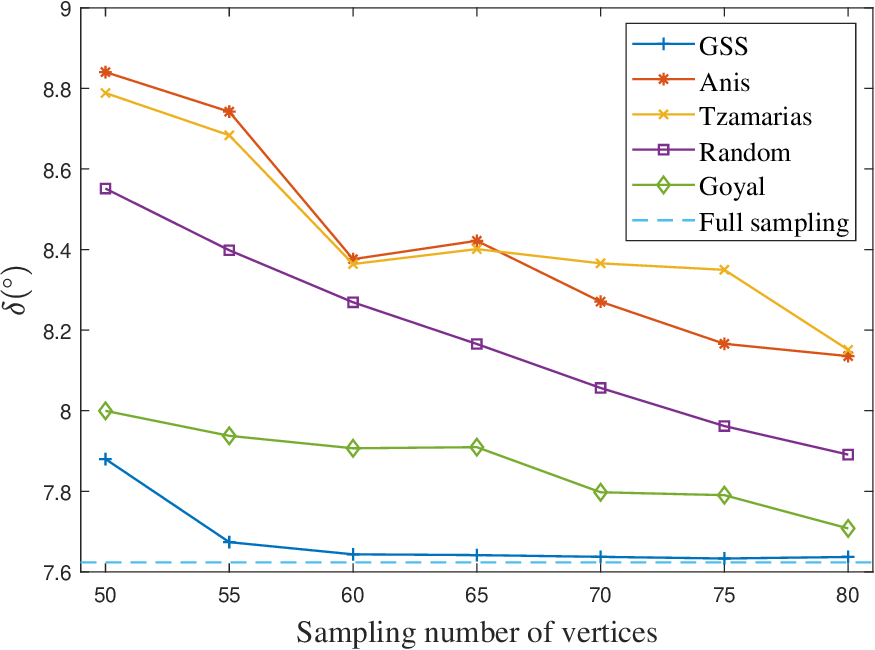}}
        \hfill
    \subfloat[]{
        \label{fig:10b}
        \includegraphics[width=0.7\linewidth]{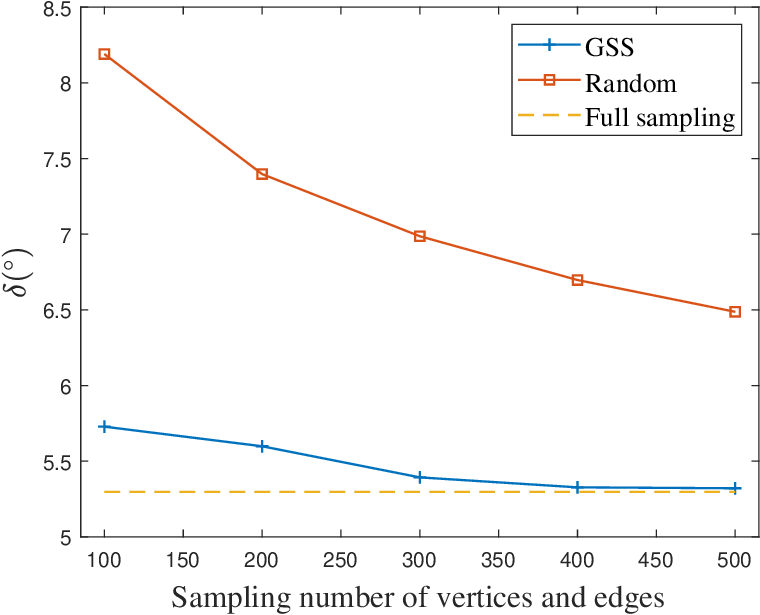}}
    \caption{ 
    Comparison of recovery performance in terms of $\delta$ among sampling algorithms for different sampling numbers on the rating dataset in \ref{realdata}. In (a), only vertices are sampled, and in (b), both vertices and edges are sampled.
    }
    \label{fig:10}
\end{figure}

The recovery quality can be evaluated by two metrics: $\delta$ in (\ref{error in angle}) and classification accuracy. 
We consider top1 and top2 classification cases. 
Top1 accuracy refers to the proportion of test instances,  whose respective classification category with the highest confidence is consistent with the label. 
Top2 accuracy refers to the proportion of test instances,  whose respective two classification categories with the highest confidence contain the label. 
 
As shown in \figref{fig:10}, the GSS algorithm consistently outperforms most benchmarks for the recovery task, in both cases that only vertices can be sampled and that vertices and edges can be sampled simultaneously. In addition, no matter the whole sample size is $|\mathcal{V}| = 100$ or $|\mathcal{V}|+|\mathcal{E}|=4696$, our approach has a good approximation of the ideal case (``Full sampling'').

\begin{figure}[h]
    \centering
    \includegraphics[width=0.7\linewidth]{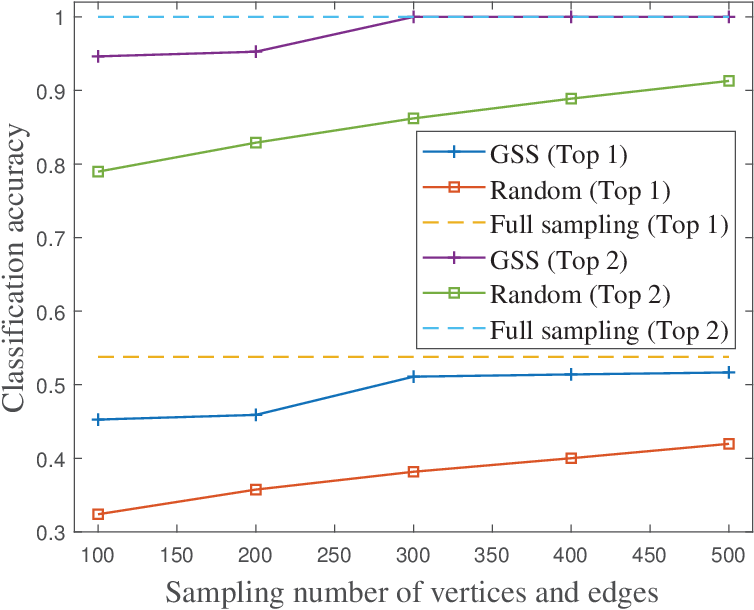}
    \caption{\label{fig:classification} Comparison of recovery performance in terms of the classification accuracy between sampling algorithms for different sampling numbers on the rating dataset in \ref{realdata}.}
\end{figure}

On the other hand, for the case that both vertices and edges can be sampled, the classification accuracy of our approach is also evaluated. The results presented in \figref{fig:classification} illustrate that the GSS algorithm offers much better performance than random sampling, and is comparable to the ideal case for a much smaller number of samples.

\section{Conclusion}
\label{sec:conclusion}
In this paper, we consider the scenario in that only the sign information of a bandlimited graph signal is observed. 
To perform sampling and recovery of the direction, we treat the sampling process as an online decision problem under the framework of MDP.
A greedy sampling algorithm and a recovery algorithm are proposed, which provide an efficient 
 and novel technical route for handling the scenarios with sign information. The methods are justified by theoretical results.
To demonstrate the effectiveness of the proposed methods, We conduct experiments including recovery and classification tasks with both synthetic and real datasets. In our future work, we try to solve the signed sampling problem with noise.


{\appendices

\section{Proof of the Lemma \ref{branch}}
\label{A-B}
Consider two partial realizations $O_{t_1}=(a_{0:t_1-1},y_{0:t_1-1})$, $O_{t_2}=(a_{0:t_2-1},y_{0:t_2-1})$, where $O_{t_1}$ is a subrealization of $O_{t_2}$.  
By (\ref{sign_space_closed}) and (\ref{eq:hmh}), $\hat{\mathcal{J}}(O_{t_1})$ and $\hat{\mathcal{J}}(O_{t_2})$ have the form
\begin{equation*}
\begin{aligned}
    \hat{\mathcal{J}}(O_{t_1}) = \left(\bigcap_{i=0}^{t_1}\{\boldsymbol{w} \mid\ \boldsymbol{\psi }_{a_i}^\top\boldsymbol{U}_B\boldsymbol{w}\lesseqqgtr y_i\}\right) \bigcap \mathcal{B}_B,\\
     \hat{\mathcal{J}}(O_{t_2}) = \left(\bigcap_{i=0}^{t_2}\{\boldsymbol{w} \mid\ \boldsymbol{\psi }_{a_i}^\top\boldsymbol{U}_B\boldsymbol{w}\lesseqqgtr y_i\}\right) \bigcap \mathcal{B}_B,
\end{aligned}
\end{equation*}
where $\mathcal{B}_B$ is the unit ball in $\mathbb{R}^B$ as in (\ref{body}), i.e. $\mathcal{B}_B = \{\boldsymbol{w} \in \mathbb{R}^B \mid \|\boldsymbol{w}\|_2 \leq 1 \}$. 
Obviously, we have $\hat{\mathcal{J}}(O_{t_2}) \subseteq \hat{\mathcal{J}}(O_{t_1})$. Therefore, $\varphi(O_{t_2}) \geq \varphi(O_{t_1})$, which proves the monotonicity. In other words, a shorter observation sequence corresponds to a state with a larger volume, because the feasible region is constructed from fewer constraints. 

 \section{Proof of the Lemma \ref{submodular}}
\label{A-C}

Firstly, we show that $\phi(O_t) \geq 0$. Note that $\hat{\mathcal{J}}_0$ can be viewed as $\hat{\mathcal{J}}(\emptyset)$. According to the monotonicity in Lemma \ref{branch}, for any partial realization $O_t$ with $a_{0:t-1}\not=\emptyset$, we have $\text{Vol}(\hat{\mathcal{J}}_0)\geq\text{Vol}(\hat{\mathcal{J}}(O_t))$, because $\hat{\mathcal{J}}_0$ is constructed by fewer constraints. So, $\phi(O_t) \geq 0$ holds.

For any partial realization $O_{t}=(a_{0:t-1},y_{0:t-1})$ (already observed), and $a \in \mathcal{A}\backslash a_{0:t-1}$ with sign $y_a$, the conditional expected marginal benefit of $a$ can be written as
\begin{equation*}
    \begin{aligned}
        &\Delta(a\mid O_{t})\\
        =&  \mathbb{E}\left[\phi(O_{t}\cup\{(a,y_a)\})-\phi(O_{t})\mid O_{t}\right]\\
         =&  \mathbb{E}\left[\text{Vol}\left(\hat{\mathcal{J}}(O_{t})\right) - \text{Vol}\left(\hat{\mathcal{J}}(O_{t}\cup \{(a,y_a)\})\right)\mid O_{t} \right].
    \end{aligned}
\end{equation*}
Whatever the sign of $a$ is, according to the monotonicity in Lemma \ref{branch}, we have $\phi(O_{t}\cup \{(a,y_a)\}) \geq \phi(O_{t})$ for all possible $y_a$. Thus, the expectation is non-negative, i.e., $\Delta(a\mid O_{t}) \geq 0$. The adaptive monotonicity holds.

Take two partial realizations $O_{t_1}=(a_{0:t_1-1},y_{0:t_1-1})$ and $O_{t_2}=(a_{0:t_2-1},y_{0:t_2-1})$ ($O_{t_1}$ is a subrealization of $O_{t_2}$). 
The conditional expected marginal benefit $\Delta(a\mid O_{t_1})$ of some $a\in \mathcal{A}\backslash\ a_{0:t_2-1}$ can be further written as
\begin{equation*}
    \begin{aligned}
        & \Delta(a\mid O_{t_1})\\
        = & \mathbb{E}\left[\text{Vol}\left(\hat{\mathcal{J}}(O_{t_1})\right) - \text{Vol}\left(\hat{\mathcal{J}}(O_{t_1}\cup \{(a,y_a)\})\right)\mid O_{t_1} \right]\\
        =&  \text{Vol}\left(\hat{\mathcal{J}}(O_{t_1})\right) - \mathbb{E}\left[ \text{Vol}\left(\hat{\mathcal{J}}(O_{t_1}\cup \{(a,y_a)\})\right)\mid O_{t_1} \right].
    \end{aligned}
\end{equation*}
According to the definition of the transition probability in Section~\ref{problem_MDP} and the fact that $\text{Vol}\left(\hat{\mathcal{J}}(O_{t_1}\cup \{(a,+)\})\right)+\text{Vol}\left(\hat{\mathcal{J}}(O_{t_1}\cup \{(a,-)\})\right)=\text{Vol}\left(\hat{\mathcal{J}}(O_{t_1})\right)$, we have
    \begin{align*}
        &\Delta(a\mid O_{t_1})  \\
        =& \text{Vol}\left(\hat{\mathcal{J}}(O_{t_1})\right) - \mathbb{E}\left[ \text{Vol}\left(\hat{\mathcal{J}}(O_{t_1}\cup \{(a,y_a)\})\right)\mid O_{t_1} \right]\\
         =& \text{Vol}\left(\hat{\mathcal{J}}({O}_{t_1})\right) -\frac{\text{Vol}\left(\hat{\mathcal{J}}(O_{t_1}\cup \{(a,+)\})\right)^2}{\text{Vol}\left(\hat{\mathcal{J}}({O}_{t_1})\right)}
         \\
        & -\frac{\text{Vol}\left(\hat{\mathcal{J}}(O_{t_1}\cup \{(a,-)\})\right)^2}{\text{Vol}\left(\hat{\mathcal{J}}({O}_{t_1})\right)}\\
         =& \frac{2}{\frac{1}{\text{Vol}\left(\hat{\mathcal{J}}(O_{t_1}\cup \{(a,+)\})\right)}+\frac{1}{\text{Vol}\left(\hat{\mathcal{J}}(O_{t_1}\cup \{(a,-)\})\right)}}.
    \end{align*}
Similarly, 
$\Delta(a\mid O_{t_2})$ can be written as
\begin{equation*}
        \Delta(a\mid O_{t_2}) = \frac{2}{\frac{1}{\text{Vol}\left(\hat{\mathcal{J}}(O_{t_2}\cup \{(a,+)\})\right)}+\frac{1}{\text{Vol}\left(\hat{\mathcal{J}}(O_{t_2}\cup \{(a,-)\})\right)}}.
\end{equation*}

According to Lemma \ref{branch}, since $O_{t_1} \cup \{(a,+)\}$, $O_{t_1} \cup \{(a,-)\}$ are a subrealization of $O_{t_2}\cup \{(a,+)\}$, $O_{t_2}\cup \{(a,-)\}$ respectively, we have 
\begin{equation*}
\begin{aligned}
    \text{Vol}\left(\hat{\mathcal{J}}(O_{t_1}\cup \{(a,+)\})\right) &\geq \text{Vol}\left(\hat{\mathcal{J}}(O_{t_2}\cup \{(a,+)\})\right),\\
    \text{Vol}\left(\hat{\mathcal{J}}(O_{t_1}\cup \{(a,-)\})\right) &\geq \text{Vol}\left(\hat{\mathcal{J}}(O_{t_2}\cup \{(a,-)\})\right).
\end{aligned}
\end{equation*}

Thus, $\Delta(a\mid O_{t_1}) \geq \Delta(a\mid O_{t_2})$. The adaptive submodularity holds.

\section{Proof of the Theorem \ref{theorem1}}
\label{A-A}

\begin{lemma}
  $\boldsymbol{P}$ is a firmly non-expansive operator.  
\end{lemma}
\begin{proof}

Combining (\ref{Pv}) and (\ref{Pe}), the projection can be written in a unified form 
\begin{equation*}
    \boldsymbol{P}\boldsymbol{w}:=\left\{ \begin{array}{lr}
    \boldsymbol{w} - \frac{\boldsymbol{u}\boldsymbol{u}^\top}{\|\boldsymbol{u}\|^2} \boldsymbol{w}, &\quad \text{if some sign is inconsistent}; 
    \\
      \boldsymbol{w}, &  \text{otherwise},
    \end{array}
    \right.
\end{equation*}
where $\boldsymbol{u}$ is a $B\times 1$ vector, corresponding to $(\boldsymbol{U}_B)_j^\top$ in (\ref{Pv}) and $(\boldsymbol{U}_B)_p - (\boldsymbol{U}_B)_q$ in (\ref{Pe}). Let $\tilde{\boldsymbol{P}} = \frac{\boldsymbol{u}\boldsymbol{u}^\top}{\|\boldsymbol{u}\|^2}$, then 
$\boldsymbol{P} = \boldsymbol{I} - \tilde{\boldsymbol{P}}$. Note that $\tilde{\boldsymbol{P}}$ is a standard projection operator in Euclidean space, and so does $\boldsymbol{P}$. Obviously, $\boldsymbol{P}$ has only eigenvalues of $0$ and $1$, which means $2\boldsymbol{P}-\boldsymbol{I}$  has only eigenvalues of $-1$ and $1$, i.e., $\|2\boldsymbol{P}-\boldsymbol{I}\|=1$. Then for any $\boldsymbol{w}\in \mathbb{R}^B$, $\|(2\boldsymbol{P}-\boldsymbol{I})\boldsymbol{w}\|\leq  \|\boldsymbol{w}\|$. 

Therefore, $(2\boldsymbol{P}-\boldsymbol{I})$ is non-expansive, and $\boldsymbol{P}$ is firmly non-expansive \cite[Fact 1.3]{bauschke1996projection}.

\end{proof}

\begin{lemma}
  $\{\hat{\mathcal{C}}\}$ is boundedly linearly regular.
\end{lemma}
\begin{proof}
  For $\hat{\mathcal{C}}$, it is a closed convex cone in a finite-dimensional Hilbert space.
  By Proposition 5.4 and 5.9 in \cite{bauschke1996projection}, we can derive that the $\{\hat{\mathcal{C}}\}$ is boundedly linearly regular.
\end{proof}
  
  In step 4 of Algorithm \ref{alg:1}, each iteration is made up of a firmly non-expansive operation.
  According to the algorithm settings in \cite{bauschke1996projection}, it is not difficult to verify that this algorithm is cyclic, singular, and unrelaxed.
Furthermore, referring to Definition 4.8 in \cite{bauschke1996projection}, the algorithm is linearly focusing.
  
  With the help of Corollary 3.12 in \cite{bauschke1996projection}, it is proved that $\left\{\boldsymbol{h}_\mathit{n}\right\}$ converges linearly to some point in $\hat{\mathcal{C}}$.
  According to Theorem 5.7 in \cite{bauschke1996projection}, we can further conclude that the convergence rate is independent of 
  $\boldsymbol{h}_0$.

  \section{Proof of the Theorem \ref{main conclusion}}
  \label{A-D}
  Firstly, we show that the solution to Problem~(\ref{optimization}) can be well approximated by the greedy algorithm.
  
  Note that, the conditional expectation of $\phi$ over $O_t$ in Lemma \ref{submodular} conditioned on $\hat{\mathcal{J}}_0$ is exactly the objective function in (\ref{optimization}), i.e., $\mathbb{E}[\phi(O_t)\mid \hat{\mathcal{J}}_0]=\mathbb{E}[\phi(a_{0:t-1},y_{0:t-1})\mid \hat{\mathcal{J}}_0]=\Phi(a_{0:t-1})$. By the adaptive monotonicity and adaptive submodularity, the greedy approach can provide a 
  solution to Problem~(\ref{optimization}) that is comparable to the global optimum \cite{golovin2011adaptive,krause2014submodular}.  

To be specific, we maximize the expected reward for each step. Therefore, at each time step $t$, we can use the following greedy sampling scheme to solve the sampling sequence as
\begin{align}
    \label{greedy optimization}
        a_t^* &= \mathop{\text{argmax}}\limits_{a_t\in \mathcal{A}_t}\ \Delta(a_t \mid O_t^*)\nonumber \\
        &= \mathop{\text{argmax}}\limits_{a_t\in \mathcal{A}_t}\ -\mathbb{E}\left[\text{Vol}\left(\hat{\mathcal{J}}(O_t^* \cup \{(a_t,y_t)\})\right) \mid O_t^*\right].
\end{align}
In (\ref{greedy optimization}), the notations are consistent with (\ref{equal volume}). It has been proved in \cite{golovin2011adaptive} that the sampling sequence solved in an online manner by (\ref{greedy optimization}) can achieve the performance within $(1-\frac{1}{e})$ of the globally optimal solution of (\ref{optimization}).

Next, we show that (\ref{equal volume}) is equivalent to (\ref{greedy optimization}). 
Consider the cases $y_t=+$ and $y_t=-$, 
then we can rewrite (\ref{greedy optimization}) as
\begin{equation*}
\label{concrete optimization}
  a_t^*  = \mathop{\text{argmin}}\limits_{a_t\in \mathcal{A}_t}\ \sum_{y_t=+,-}\mathbb{P}(y_t\mid O_t^*)\text{Vol}\left(\hat{\mathcal{J}}(O_t^*\cup \{(a_t,y_t)\})\right),
\end{equation*}
where $\mathbb{P}(y_t\mid O_t^*)$ is defined in (\ref{prob_yt}).
 It is obvious that $\text{Vol}(\hat{\mathcal{J}}(O_t^*))=\text{Vol}\left(\hat{\mathcal{J}}(O_t^*\cup \{(a_t,+)\})\right) + \text{Vol}\left(\hat{\mathcal{J}}(O_t^*\cup \{(a_t,-)\})\right) $, we have

\begin{equation*}
    \begin{aligned}
        &\frac{\text{Vol}\left(\hat{\mathcal{J}}(O_t^*\cup \{(a_t,+)\})\right)^2}{\text{Vol}(\hat{\mathcal{J}}(O_t^*))} + \frac{\text{Vol}\left(\hat{\mathcal{J}}(O_t^*\cup \{(a_t,-)\})\right)^2}{\text{Vol}(\hat{\mathcal{J}}(O_t^*))}\\
    &\geq \frac{\text{Vol}\left(\hat{\mathcal{J}}(O_t^*\cup \{(a_t,+)\})\right)+\text{Vol}\left(\hat{\mathcal{J}}(O_t^*\cup \{(a_t,-)\})\right)}{2}\\
    & =\frac{\text{Vol}(\hat{\mathcal{J}}(O_t^*))}{2}.
    \end{aligned}
\end{equation*}
The LHS of the above inequality has a minimum (RHS) that is independent of $a_t$. For two real numbers with a constant sum, the smaller their difference, the smaller their square sum. So in order to minimize the LHS, we can minimize $\left|\text{Vol}\left(\hat{\mathcal{J}}(O_t^*\cup \{(a_t,+)\})\right) - \text{Vol}\left(\hat{\mathcal{J}}(O_t^*\cup \{(a_t,-)\})\right)\right|$. 
Therefore, (\ref{greedy optimization}) can be replaced equivalently by (\ref{equal volume}).


 }


\bibliographystyle{IEEEtran}
\bibliography{IEEEabrv,ref}

\begin{thebibliography}{10}
\providecommand{\url}[1]{#1}
\csname url@samestyle\endcsname
\providecommand{\newblock}{\relax}
\providecommand{\bibinfo}[2]{#2}
\providecommand{\BIBentrySTDinterwordspacing}{\spaceskip=0pt\relax}
\providecommand{\BIBentryALTinterwordstretchfactor}{4}
\providecommand{\BIBentryALTinterwordspacing}{\spaceskip=\fontdimen2\font plus
\BIBentryALTinterwordstretchfactor\fontdimen3\font minus \fontdimen4\font\relax}
\providecommand{\BIBforeignlanguage}[2]{{%
\expandafter\ifx\csname l@#1\endcsname\relax
\typeout{** WARNING: IEEEtran.bst: No hyphenation pattern has been}%
\typeout{** loaded for the language `#1'. Using the pattern for}%
\typeout{** the default language instead.}%
\else
\language=\csname l@#1\endcsname
\fi
#2}}
\providecommand{\BIBdecl}{\relax}
\BIBdecl

\bibitem{liu2022recovery}
W.~Liu, H.~Feng, K.~Wang, F.~Ji, and B.~Hu, ``Recovery of graph signals from sign measurements,'' in \emph{ICASSP 2022-2022 IEEE International Conference on Acoustics, Speech and Signal Processing (ICASSP)}.\hskip 1em plus 0.5em minus 0.4em\relax IEEE, 2022, pp. 5927--5931.

\bibitem{stankovic2019understanding}
L.~Stankovic, D.~P. Mandic, M.~Dakovic, I.~Kisil, E.~Sejdic, and A.~G. Constantinides, ``Understanding the basis of graph signal processing via an intuitive example-driven approach [lecture notes],'' \emph{IEEE Signal Processing Magazine}, vol.~36, no.~6, pp. 133--145, 2019.

\bibitem{zhou2004regularization}
D.~Zhou and B.~Sch{\"o}lkopf, ``A regularization framework for learning from graph data,'' in \emph{ICML 2004 Workshop on Statistical Relational Learning and Its Connections to Other Fields (SRL 2004)}, 2004, pp. 132--137.

\bibitem{zhang2020deep}
Z.~Zhang, P.~Cui, and W.~Zhu, ``Deep learning on graphs: A survey,'' \emph{IEEE Transactions on Knowledge and Data Engineering}, 2020.

\bibitem{zhang2024graph}
Y.~Zhang, J.~Li, J.~Ding, and X.~Li, ``A graph transformer-driven approach for network robustness learning,'' \emph{IEEE Transactions on Circuits and Systems I: Regular Papers}, 2024.

\bibitem{shuman2013emerging}
D.~I. Shuman, S.~K. Narang, P.~Frossard, A.~Ortega, and P.~Vandergheynst, ``The emerging field of signal processing on graphs: Extending high-dimensional data analysis to networks and other irregular domains,'' \emph{IEEE signal processing magazine}, vol.~30, no.~3, pp. 83--98, 2013.

\bibitem{stankovic2020data}
L.~Stankovi{\'c}, D.~Mandic, M.~Dakovi{\'c}, M.~Brajovi{\'c}, B.~Scalzo, S.~Li, A.~G. Constantinides \emph{et~al.}, ``Data analytics on graphs part ii: Signals on graphs,'' \emph{Foundations and Trends{\textregistered} in Machine Learning}, vol.~13, no. 2-3, 2020.

\bibitem{xia2021graph}
F.~Xia, K.~Sun, S.~Yu, A.~Aziz, L.~Wan, S.~Pan, and H.~Liu, ``Graph learning: A survey,'' \emph{IEEE Transactions on Artificial Intelligence}, vol.~2, no.~2, pp. 109--127, 2021.

\bibitem{cvetkovic2009applications}
D.~M. Cvetkovic, ``Applications of graph spectra: An introduction to the literature,'' \emph{Appl. Graph Spectra}, vol.~13, no.~21, pp. 7--31, 2009.

\bibitem{gavili2017shift}
A.~Gavili and X.-P. Zhang, ``On the shift operator, graph frequency, and optimal filtering in graph signal processing,'' \emph{IEEE Transactions on Signal Processing}, vol.~65, no.~23, pp. 6303--6318, 2017.

\bibitem{wang2015generalized}
X.~Wang, J.~Chen, and Y.~Gu, ``Generalized graph signal sampling and reconstruction,'' in \emph{2015 IEEE Global Conference on Signal and Information Processing (GlobalSIP)}.\hskip 1em plus 0.5em minus 0.4em\relax IEEE, 2015, pp. 567--571.

\bibitem{ortega2018graph}
A.~Ortega, P.~Frossard, J.~Kova{\v{c}}evi{\'c}, J.~M. Moura, and P.~Vandergheynst, ``Graph signal processing: Overview, challenges, and applications,'' \emph{Proceedings of the IEEE}, vol. 106, no.~5, pp. 808--828, 2018.

\bibitem{tanaka2020sampling}
Y.~Tanaka, Y.~C. Eldar, A.~Ortega, and G.~Cheung, ``Sampling signals on graphs: From theory to applications,'' \emph{IEEE Signal Processing Magazine}, vol.~37, no.~6, pp. 14--30, 2020.

\bibitem{tsitsvero2016signals}
M.~Tsitsvero, S.~Barbarossa, and P.~Di~Lorenzo, ``Signals on graphs: Uncertainty principle and sampling,'' \emph{IEEE Transactions on Signal Processing}, vol.~64, no.~18, pp. 4845--4860, 2016.

\bibitem{marques2015sampling}
A.~G. Marques, S.~Segarra, G.~Leus, and A.~Ribeiro, ``Sampling of graph signals with successive local aggregations,'' \emph{IEEE Transactions on Signal Processing}, vol.~64, no.~7, pp. 1832--1843, 2015.

\bibitem{xie2017design}
X.~Xie, H.~Feng, J.~Jia, and B.~Hu, ``Design of sampling set for bandlimited graph signal estimation,'' in \emph{2017 IEEE Global Conference on Signal and Information Processing (GlobalSIP)}.\hskip 1em plus 0.5em minus 0.4em\relax IEEE, 2017, pp. 653--657.

\bibitem{fang2011smart}
X.~Fang, S.~Misra, G.~Xue, and D.~Yang, ``Smart grid—the new and improved power grid: A survey,'' \emph{IEEE communications surveys \& tutorials}, vol.~14, no.~4, pp. 944--980, 2011.

\bibitem{falk2019practical}
K.~Falk, \emph{Practical recommender systems}.\hskip 1em plus 0.5em minus 0.4em\relax Simon and Schuster, 2019.

\bibitem{bansal2022universal}
S.~Bansal, A.~Bhattacharyya, A.~Chaturvedi, and J.~Scarlett, ``Universal 1-bit compressive sensing for bounded dynamic range signals,'' in \emph{2022 IEEE International Symposium on Information Theory (ISIT)}.\hskip 1em plus 0.5em minus 0.4em\relax IEEE, 2022, pp. 3280--3284.

\bibitem{beheshti2022adaptive}
H.~Beheshti, S.~Daei, and F.~Haddadi, ``Adaptive recovery of dictionary-sparse signals using binary measurements,'' \emph{EURASIP Journal on Advances in Signal Processing}, vol. 2022, no.~1, pp. 1--13, 2022.

\bibitem{jacques2013robust}
L.~Jacques, J.~N. Laska, P.~T. Boufounos, and R.~G. Baraniuk, ``Robust 1-bit compressive sensing via binary stable embeddings of sparse vectors,'' \emph{IEEE Transactions on Information Theory}, vol.~59, no.~4, pp. 2082--2102, 2013.

\bibitem{zeng2022one}
Y.~Zeng, S.~Khobahi, and M.~Soltanalian, ``One-bit compressive sensing: Can we go deep and blind?'' \emph{IEEE Signal Processing Letters}, vol.~29, pp. 1629--1633, 2022.

\bibitem{tachella2023learning}
J.~Tachella and L.~Jacques, ``Learning to reconstruct signals from binary measurements,'' \emph{arXiv preprint arXiv:2303.08691}, 2023.

\bibitem{eamaz2022one}
A.~Eamaz, F.~Yeganegi, and M.~Soltanalian, ``One-bit phase retrieval: More samples means less complexity?'' \emph{IEEE Transactions on Signal Processing}, vol.~70, pp. 4618--4632, 2022.

\bibitem{kishore2020wirtinger}
V.~Kishore and C.~S. Seelamantula, ``Wirtinger flow algorithms for phase retrieval from binary measurements,'' in \emph{ICASSP 2020-2020 IEEE International Conference on Acoustics, Speech and Signal Processing (ICASSP)}.\hskip 1em plus 0.5em minus 0.4em\relax IEEE, 2020, pp. 5750--5754.

\bibitem{khobahi2018signal}
S.~Khobahi and M.~Soltanalian, ``Signal recovery from 1-bit quantized noisy samples via adaptive thresholding,'' in \emph{2018 52nd Asilomar Conference on Signals, Systems, and Computers}.\hskip 1em plus 0.5em minus 0.4em\relax IEEE, 2018, pp. 1757--1761.

\bibitem{goyal2018estimation}
M.~Goyal and A.~Kumar, ``Estimation of bandlimited signals on graphs from single bit recordings of noisy samples,'' in \emph{2018 26th European Signal Processing Conference (EUSIPCO)}.\hskip 1em plus 0.5em minus 0.4em\relax IEEE, 2018, pp. 902--906.

\bibitem{logan1977information}
B.~F. Logan~Jr, ``Information in the zero crossings of bandpass signals,'' \emph{Bell System Technical Journal}, vol.~56, no.~4, pp. 487--510, 1977.

\bibitem{rotem1986image}
D.~Rotem and Y.~Zeevi, ``Image reconstruction from zero crossings,'' \emph{IEEE Transactions on Acoustics, Speech, and Signal Processing}, vol.~34, no.~5, pp. 1269--1277, 1986.

\bibitem{oppenheim1981importance}
A.~V. Oppenheim and J.~S. Lim, ``The importance of phase in signals,'' \emph{Proceedings of the IEEE}, vol.~69, no.~5, pp. 529--541, 1981.

\bibitem{dhar2019tests}
S.~S. Dhar, D.~Kundu, and U.~Das, ``Tests for the parameters of chirp signal model,'' \emph{IEEE Transactions on Signal Processing}, vol.~67, no.~16, pp. 4291--4301, 2019.

\bibitem{wu2008subspace}
Y.~Wu, H.~C. So, and H.~Liu, ``Subspace-based algorithm for parameter estimation of polynomial phase signals,'' \emph{IEEE Transactions on Signal Processing}, vol.~56, no.~10, pp. 4977--4983, 2008.

\bibitem{chamon2016near}
L.~F. Chamon and A.~Ribeiro, ``Near-optimality of greedy set selection in the sampling of graph signals,'' in \emph{2016 IEEE Global Conference on Signal and Information Processing (GlobalSIP)}.\hskip 1em plus 0.5em minus 0.4em\relax IEEE, 2016, pp. 1265--1269.

\bibitem{wu2018greedy}
C.~Wu, W.~Chen, and J.~Zhang, ``Greedy algorithm with approximation ratio for sampling noisy graph signals,'' in \emph{2018 IEEE International Conference on Acoustics, Speech and Signal Processing (ICASSP)}.\hskip 1em plus 0.5em minus 0.4em\relax IEEE, 2018, pp. 4654--4658.

\bibitem{anis2014towards}
A.~Anis, A.~Gadde, and A.~Ortega, ``Towards a sampling theorem for signals on arbitrary graphs,'' in \emph{2014 IEEE International Conference on Acoustics, Speech and Signal Processing (ICASSP)}.\hskip 1em plus 0.5em minus 0.4em\relax IEEE, 2014, pp. 3864--3868.

\bibitem{gadde2015probabilistic}
A.~Gadde and A.~Ortega, ``A probabilistic interpretation of sampling theory of graph signals,'' in \emph{2015 IEEE international conference on Acoustics, Speech and Signal Processing (ICASSP)}.\hskip 1em plus 0.5em minus 0.4em\relax IEEE, 2015, pp. 3257--3261.

\bibitem{jayawant2018distance}
A.~Jayawant and A.~Ortega, ``A distance-based formulation for sampling signals on graphs,'' in \emph{2018 IEEE International Conference on Acoustics, Speech and Signal Processing (ICASSP)}.\hskip 1em plus 0.5em minus 0.4em\relax IEEE, 2018, pp. 6318--6322.

\bibitem{tzamarias2018novel}
D.~E.~O. Tzamarias, P.~Akyazi, and P.~Frossard, ``A novel method for sampling bandlimited graph signals,'' in \emph{2018 26th European Signal Processing Conference (EUSIPCO)}.\hskip 1em plus 0.5em minus 0.4em\relax IEEE, 2018, pp. 126--130.

\bibitem{lin2019active}
S.~Lin, X.~Xie, H.~Feng, and B.~Hu, ``Active sampling for approximately bandlimited graph signals,'' in \emph{ICASSP 2019-2019 IEEE International Conference on Acoustics, Speech and Signal Processing (ICASSP)}.\hskip 1em plus 0.5em minus 0.4em\relax IEEE, 2019, pp. 5441--5445.

\bibitem{berberidis2017active}
D.~Berberidis and G.~B. Giannakis, ``Active sampling for graph-aware classification,'' in \emph{2017 IEEE Global Conference on Signal and Information Processing (GlobalSIP)}.\hskip 1em plus 0.5em minus 0.4em\relax IEEE, 2017, pp. 648--652.

\bibitem{barker1973lattice}
G.~P. Barker, ``The lattice of faces of a finite dimensional cone,'' \emph{Linear Algebra and its Applications}, vol.~7, no.~1, pp. 71--82, 1973.

\bibitem{bather2000decision}
J.~Bather, \emph{Decision theory: An introduction to dynamic programming and sequential decisions}.\hskip 1em plus 0.5em minus 0.4em\relax John Wiley \& Sons, Inc., 2000.

\bibitem{russell2010artificial}
S.~J. Russell, \emph{Artificial intelligence a modern approach}.\hskip 1em plus 0.5em minus 0.4em\relax Pearson Education, Inc., 2010.

\bibitem{simonovits2003compute}
M.~Simonovits, ``How to compute the volume in high dimension?'' \emph{Mathematical programming}, vol.~97, no.~1, pp. 337--374, 2003.

\bibitem{bertsimas1997introduction}
D.~Bertsimas and J.~N. Tsitsiklis, \emph{Introduction to linear optimization}.\hskip 1em plus 0.5em minus 0.4em\relax Athena Scientific Belmont, MA, 1997, vol.~6.

\bibitem{theodoridis2010adaptive}
S.~Theodoridis, K.~Slavakis, and I.~Yamada, ``Adaptive learning in a world of projections,'' \emph{IEEE Signal Processing Magazine}, vol.~28, no.~1, pp. 97--123, 2010.

\bibitem{bauschke1996projection}
H.~H. Bauschke and J.~M. Borwein, ``On projection algorithms for solving convex feasibility problems,'' \emph{SIAM review}, vol.~38, no.~3, pp. 367--426, 1996.

\bibitem{golovin2011adaptive}
D.~Golovin and A.~Krause, ``Adaptive submodularity: Theory and applications in active learning and stochastic optimization,'' \emph{Journal of Artificial Intelligence Research}, vol.~42, pp. 427--486, 2011.

\bibitem{henk2017basic}
M.~Henk, J.~Richter-Gebert, and G.~M. Ziegler, ``Basic properties of convex polytopes,'' in \emph{Handbook of discrete and computational geometry}.\hskip 1em plus 0.5em minus 0.4em\relax Chapman and Hall/CRC, 2017, pp. 383--413.

\bibitem{li2010concise}
S.~Li, ``Concise formulas for the area and volume of a hyperspherical cap,'' \emph{Asian Journal of Mathematics \& Statistics}, vol.~4, no.~1, pp. 66--70, 2010.

\bibitem{krause2014submodular}
A.~Krause and D.~Golovin, ``Submodular function maximization.'' \emph{Tractability}, vol.~3, pp. 71--104, 2014.

\end{thebibliography}


\begin{thebibliography}{1}
\bibliographystyle{IEEEtran}

\bibitem{ref1}
{\it{Mathematics Into Type}}. American Mathematical Society. [Online]. Available: https://www.ams.org/arc/styleguide/mit-2.pdf

\bibitem{ref2}
T. W. Chaundy, P. R. Barrett and C. Batey, {\it{The Printing of Mathematics}}. London, U.K., Oxford Univ. Press, 1954.

\bibitem{ref3}
F. Mittelbach and M. Goossens, {\it{The \LaTeX Companion}}, 2nd ed. Boston, MA, USA: Pearson, 2004.

\bibitem{ref4}
G. Gr\"atzer, {\it{More Math Into LaTeX}}, New York, NY, USA: Springer, 2007.

\bibitem{ref5}M. Letourneau and J. W. Sharp, {\it{AMS-StyleGuide-online.pdf,}} American Mathematical Society, Providence, RI, USA, [Online]. Available: http://www.ams.org/arc/styleguide/index.html

\bibitem{ref6}
H. Sira-Ramirez, ``On the sliding mode control of nonlinear systems,'' \textit{Syst. Control Lett.}, vol. 19, pp. 303--312, 1992.

\bibitem{ref7}
A. Levant, ``Exact differentiation of signals with unbounded higher derivatives,''  in \textit{Proc. 45th IEEE Conf. Decis.
Control}, San Diego, CA, USA, 2006, pp. 5585--5590. DOI: 10.1109/CDC.2006.377165.

\bibitem{ref8}
M. Fliess, C. Join, and H. Sira-Ramirez, ``Non-linear estimation is easy,'' \textit{Int. J. Model., Ident. Control}, vol. 4, no. 1, pp. 12--27, 2008.

\bibitem{ref9}
R. Ortega, A. Astolfi, G. Bastin, and H. Rodriguez, ``Stabilization of food-chain systems using a port-controlled Hamiltonian description,'' in \textit{Proc. Amer. Control Conf.}, Chicago, IL, USA,
2000, pp. 2245--2249.

\end{thebibliography}

%

\newpage

\vfill

\end{document}